%% file: ACC_paper_arXiv.tex
\documentclass{biometrika} 

\usepackage{amsmath} 
\usepackage{amssymb}%,amsfonts}
\usepackage{enumerate}
\usepackage{multirow}
 \usepackage[colorlinks,citecolor=blue,urlcolor=blue]{hyperref}

 \usepackage[mathscr]{eucal}
\usepackage{times}

\usepackage{graphicx} 
	\graphicspath{ {ImageFiles/} }
\usepackage{epstopdf}
\usepackage{subfig} 
\usepackage{wrapfig}

\usepackage{hyperref}
\usepackage{hypcap}

\usepackage{color}
\theoremstyle{plain}             

\newtheorem{mydef}{Definition}

\input{Symbols_pream.tex}

\newcommand{\bel}{\begin{eqnarray}\label}
\newcommand{\eel}{\end{eqnarray}}
\newcommand{\bes}{\begin{eqnarray*}}
	\newcommand{\ees}{\end{eqnarray*}}
\newcommand{\bei}{\begin{itemize}}
	\newcommand{\eei}{\end{itemize}}

\makeatletter
\newcommand{\distas}[1]{\mathbin{\overset{#1}{\kern\z@\sim}}}%
\newsavebox{\mybox}\newsavebox{\mysim}
\newcommand{\distras}[1]{%
	\savebox{\mybox}{\hbox{\kern3pt$\scriptstyle#1$\kern3pt}}%
	\savebox{\mysim}{\hbox{$\sim$}}%
	\mathbin{\overset{#1}{\kern\z@\resizebox{\wd\mybox}{\ht\mysim}{$\sim$}}}%
}
\makeatother

\definecolor{capri}{rgb}{0.0, 0.75, 1.0}
\definecolor{caribbeangreen}{rgb}{0.0, 1, 0.5}

\begin{document}

\title{%Approximate confidence distribution computing: 
An effective likelihood-free approximate computing method with statistical inferential guarantees}

\author{Suzanne Thornton}
\affil{Department of Statistics and Biostatistics, Rutgers, The State University of New Jersey, U.S.A. \email{suzanne.thornton@rutgers.edu}}

\author{Wentao Li}
\affil{ Department of Mathematics, Statistics, and Physics, Newcastle University, United Kingdom \email{wentao.li@newcastle.ac.uk}}

\author{Minge Xie}
\affil{Department of Statistics and Biostatistics, Rutgers, The State University of New Jersey, U.S.A. \email{mxie@stat.rutgers.edu}}

\maketitle

	\begin{abstract}
	Approximate Bayesian computing is a powerful likelihood-free method that has grown increasingly popular since early applications in population genetics. However, complications arise in the theoretical justification for Bayesian inference conducted from this method with a non-sufficient summary statistic. In this paper, we seek to re-frame approximate Bayesian computing within a frequentist context and justify its performance by standards set on the frequency coverage rate. In doing so, we develop a new computational technique called {\it approximate confidence distribution computing}, yielding theoretical support for the use of non-sufficient summary statistics in likelihood-free methods. Furthermore, we demonstrate that approximate confidence distribution computing extends the scope of approximate Bayesian computing to include data-dependent priors without damaging the inferential integrity. This data-dependent prior can be viewed as an initial `distribution estimate' of the target parameter which is updated with the results of the approximate confidence distribution computing method. A general strategy for constructing an appropriate data-dependent prior is also discussed and is shown to often increase the computing speed while maintaining statistical inferential guarantees. We supplement the theory with simulation studies illustrating the benefits of the proposed method, namely the potential for broader applications and the increased computing speed compared to the standard approximate Bayesian computing methods.
	\end{abstract}
	
	\begin{keywords}       
		Approximate Bayesian computing; Bernstein-von Mises; Confidence distribution; Exact inference; Large sample theory.
	\end{keywords}
%-----------------------------------------------------------------------------------------------------------------------------	
 
\section{Introduction}
\label{sec:intro}
\subsection{Background to approximate Bayesian computing}
Approximate Bayesian computing is a likelihood-free method that approximates a posterior distribution while avoiding direct calculation of the likelihood. This procedure originated in population genetics where complex demographic histories yield intractable likelihoods. Since then, approximate Bayesian computing has been applied to many other areas besides the biological sciences including astronomy and finance; cf., e.g., \cite{Cameron2012, Csillery2010, Peters2012}. Despite its practical popularity in providing a Bayesian solution for complex data problems, the theoretical justification for inference from this method is under-developed and has only recently been explored in statistical literature; cf., e.g., \cite{Robinson2014, Barber2015, Frazier2016, Li2016}. In this paper, we seek to re-frame the problem within a frequentist setting and help address two weaknesses of approximate Bayesian computing: (1) lack of  theoretical justification for Bayesian inference when using a non-sufficient summary statistic and (2) slow computing speed. We propose a novel likelihood-free method as a bridge connecting Bayesian and frequentist inferences and examine it within the context of the existing literature on approximate computing.

Let $x_{\rm obs} = \{x_1, \dots, x_n\}$ be an observed sample from some unknown distribution with density $ f(\cdot\mid\theta)$. 
Assume that the sample is observations of some data generating model, $M_{\theta}$, where $\theta \in \P \subset \R^p$ is unknown. For any given $\theta$, we know how to simulate artificial data from $M_{\theta}$. The standard accept-reject version of approximate Bayesian computing proceeds as follows:

\begin{algo}{(Accept-reject approximate Bayesian computing)} \label{alg:rejABC}
	\begin{tabbing}
		\quad 1. Simulate $\theta_{1},\ldots,\theta_{N}\sim \pi(\theta)$; \\
		\quad 2. For each $i=1,\ldots,N$, simulate $x^{(i)}=\{x_{1}^{(i)},\ldots,x_{n}^{(i)}\}$ from $M_{\theta_i}$;\\ 
		\quad 3. For each $i=1,\ldots,N$, accept $\theta_{i}$ with probability $K_{\veps}(s^{(i)}-s_{\rm obs})$, where $s_{{\rm obs}}= $ \\ \qquad $S_{n}(x_{\rm obs})$ and $ s^{(i)}=S_{n}(x^{(i)})$.
	\end{tabbing}
\end{algo}  
In the above algorithm, $\pi(\cdot)$ is a prior distribution function and the data is summarized by some low-dimension summary statistic, $S_n(\cdot)$
(e.g.,  $S_n(\cdot)$ is a mapping from the sample space in $\mathbb{R}^{n}$ to ${\cal S} \subset \mathbb{R}^{d}$ with $d \leq n$). The kernel probability $K_\veps(\cdot)$ follows the notation $K_{\veps}(u) = \veps^{-1}K(u/ \veps)$, where $K_\veps(\cdot)$ is a kernel function. We refer to $\veps$ as the {\it tolerance level} and typically assume it goes to zero. In many cases, $\veps$ is required to go to zero at a certain rate of $n$ (cf., e.g., \cite{Li2016}), but there are cases in finite sample development in which $\veps$ is independent of sample size $n$, see e.g. \cite{Barber2015}. 

The underlying distribution from which the accepted copies or draws of $\theta$ are generated in an Appropriate Bayesian computing algorithm is called the {\it approximate Bayesian computed posterior}, with the probability density,
\begin{align} 
 \pi_{\veps}(\theta\mid s_{\rm obs})=
\frac{\int_{ \cal S}\pi(\theta)f_n(s\mid\theta)K_{\veps}(s - s_{\rm obs})\,ds}{\int_{{\cal P} \times{\cal S}}\pi(\theta)f_n(s\mid\theta)K_{\veps}(s - s_{\rm obs})\,ds d\theta}, \label{ABC_approx_posterior}
\end{align}
and corresponding cumulative  distribution function denoted by $\Pi_{\veps}(\theta\mid s_{\rm obs})$. Here $f_n (  s \mid\theta)$ denotes the probability density of the summary statistic, implied by $f(x\mid\theta)$ and is typically unknown. We will refer to $f_n ( s \mid\theta)$ as an {\it s-likelihood}. Since this is a Bayesian procedure, 
% in addition to the assumption of the existence of a data-generating model, $M_{\theta}$, 
Algorithm \ref{alg:rejABC} assumes a prior distribution, $\pi(\cdot)$, on $\theta$. In the absence of prior information, the user may select a flat prior.

A common assertion is that $ \pi_{\veps}(\theta\mid s_{\rm obs})$ is close enough to the target posterior distribution, $ p(\theta\mid x) \propto \pi(\theta)f(x\mid\theta)$, e.g. \cite{Marin2011}; however, the quality of this approximation depends on the closeness of the tolerance level to zero and, more crucially for our purposes, on the choice of summary statistic $S_n(\cdot)$. Indeed, we have the following lemma: 
\begin{lemma}\label{ABClemma}  
%{\mylemm  
	Let $K(\cdot)$ be a symmetric kernel density function with $\int { u} K({ u}) d{ u} = 0$ and $\int  \|u\|^2 K({ u}) d{ u} < \infty$ where $\| \cdot \|$ is the Euclidean norm. Suppose the matrix of second derivatives of $f_n(s\mid\theta)$ is bounded with respective to $s$. Then
	\begin{align}
	\pi_{\veps}(\theta\mid s_{\rm obs}) \propto \pi(\theta)f_n(s_{{\rm obs}}\mid\theta)
	%{\int_{\cal P}\pi(\theta)f_n(s_{{\rm obs}}\mid\theta) d\theta} 
	+ O(\veps^2).
	\label{ABC2}
	\end{align}
%}
\end{lemma}
 
\noindent
Various versions of this result are known (cf., e.g., \cite{Barber2015} and \cite{Li2017}); for completeness, we provide a brief proof of Lemma~\ref{ABClemma} in the appendix. Note that, if the summary statistic $S_n(\cdot)$ is not sufficient, $f_n(s_{{\rm obs}}\mid\theta)$ can be very different from $f(x\mid\theta)$, in which case $ \pi_{\veps}(\theta\mid s_{\rm obs})$ can be a very poor approximation to the target posterior, $p(\theta \mid x)$, even if $\veps \to 0$. 

Figure~\ref{fig:ABC} provides such an example where we consider random data from a Cauchy distribution with a known scale parameter. Only the data itself is sufficient for the location parameter, $\theta$; therefore, any summary statistic, including the commonly used sample mean and median, will not be sufficient.  Figure~\ref{fig:ABC} illustrates that, without sufficiency, the posterior approximation resulting from Algorithm~\ref{alg:rejABC}, using either sample mean and sample median as $S_n(\cdot)$, will never converge to the targeted posterior distribution, thus indicating that the approximations to the target posterior can be quite poor. What's more, the two different summary statistics lead to quite different approximate Bayesian computed posteriors $\pi_{\veps}(\theta\mid s_{\rm obs})$. 
%Specifically, Figure~\ref{fig:ABC} shows~two applications: one using the sample mean and the other using the sample median as the summary statistic, each with a flat prior on the unknown location parameter. 
In neither case is the approximate Bayesian computed posterior the same as the targeted posterior distribution, regardless of sample size or the~rate~of $\veps \rightarrow 0$, including the rate typically required in the existing literature; cf.,~\citet{Li2016}.
% Additionally, t
The approximate Bayesian computed posteriors obtained using the sample mean are much flatter than those obtained using the sample median. Further details about Figure~\ref{fig:ABC} can be found in Section~\ref{sec:Cauchy}.

For this reason, inference from $ \Pi_{\veps}(\cdot\mid s_{\rm obs})$ can produce misleading results within a Bayesian context when the summary statistic used is not sufficient. Questions arise such as, if $\Pi_{\veps}(\cdot \mid s_{obs})$ is different from the target posterior distribution, can it still be used in Bayesian inference? Or, since different summary statistics can produce different approximate posterior distributions, can one or more of these distributions be used to make statistical inferences?

\begin{figure}[h]
	\centering	
	  \subfloat[Cauchy data of sample size $n=50$.]{\includegraphics[height = 2.5in, width = .49\linewidth ]{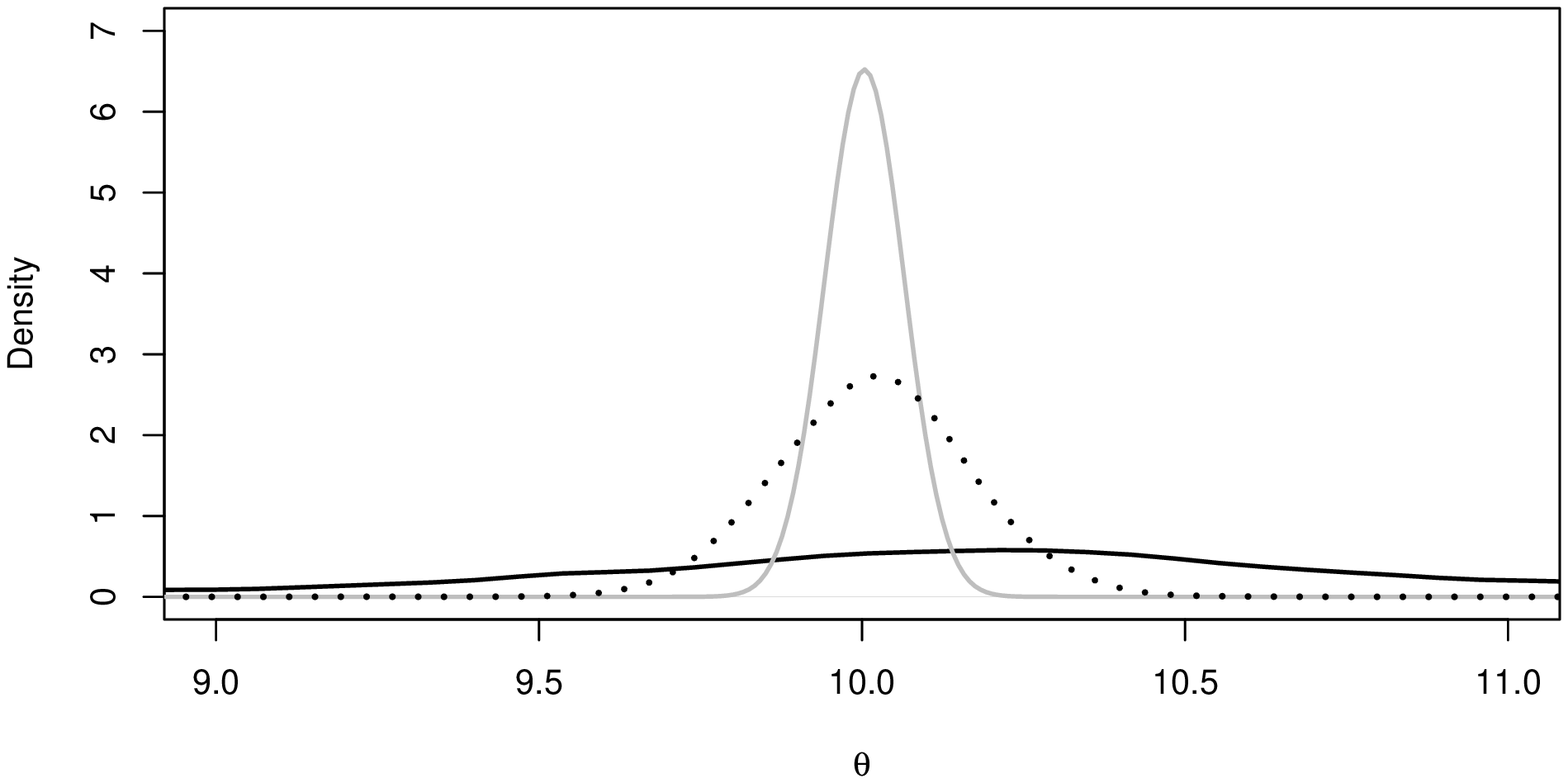}} \hfill
	  \subfloat[Cauchy data of sample size $n=5000$.]{\includegraphics[height = 2.5in, width = .49\linewidth ]{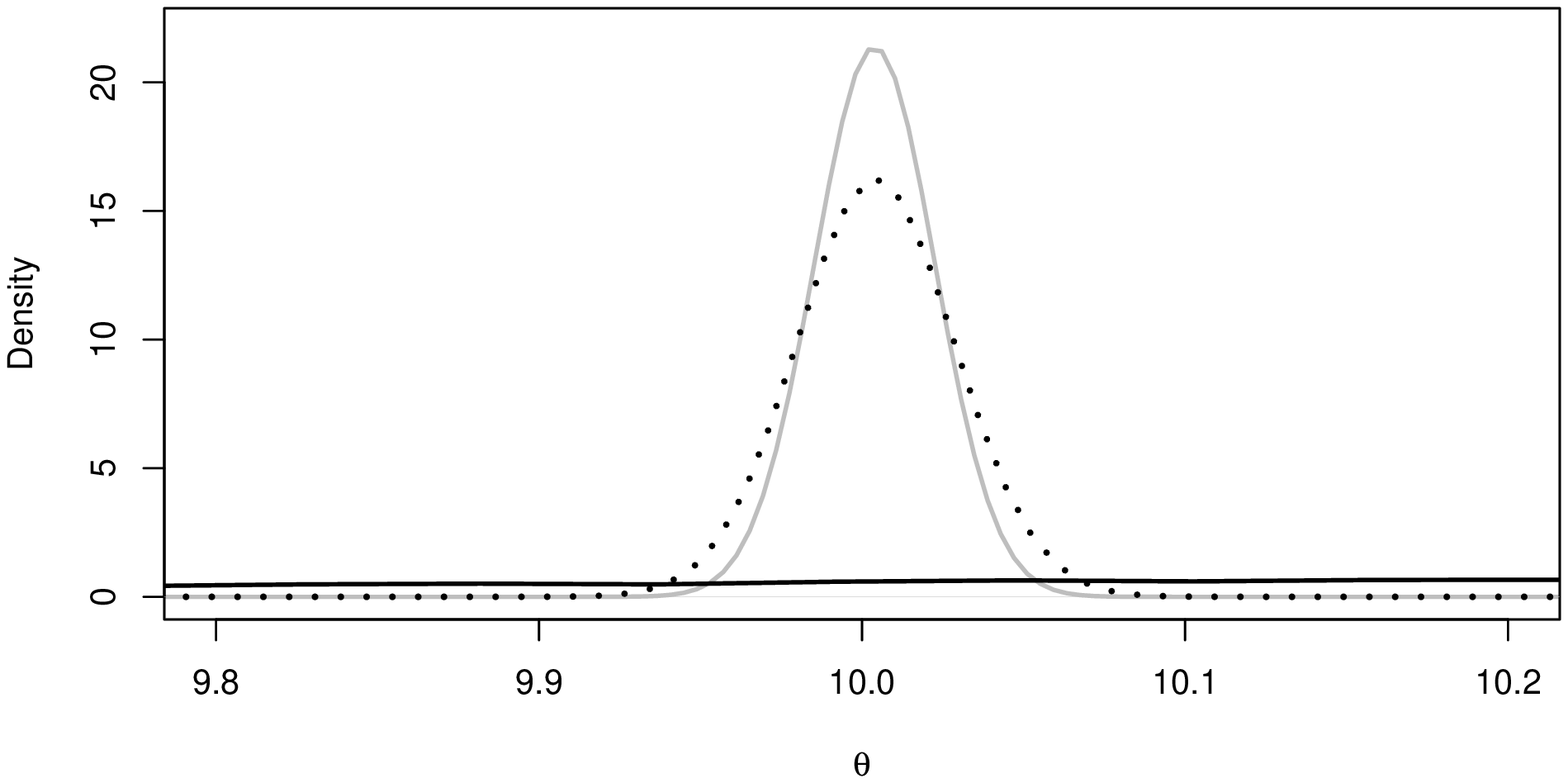}}
	\caption{{\it The three curves in each of the two plots are the target posterior (gray) and approximate Bayesian computed posteriors for data from a Cauchy distribution with known scale parameter for summary statistic $S_n = \bar{x}$ (solid black) and $S_n = \mbox{Median}(x)$ (dashed black). The prior density is a constant in $\mathbb{R}$.}
}
	\label{fig:ABC}
\end{figure}	
In this paper, we attempt to address these questions by instead re-framing Algorithm \ref{alg:rejABC} within a frequentist context, thus creating a more general  likelihood-free method based on confidence distribution theory. To this end, we introduce a new computational method called {\it approximate confidence-distribution computing}. 

\subsection{Approximate confidence distribution computing} 
When estimating an unknown parameter, we often desire that our estimators, whether point estimators or interval estimators, have certain properties such as unbiasedness or a certain coverage of the true parameter value in the long run. A confidence distribution is an extension of this tradition in that it is a distribution estimate (i.e., it uses a sample-dependent distribution function to estimate the target parameter) that satisfies certain desirable properties. Following \cite{Xie2013}, \cite{Schweder2016}, we define a confidence distribution as follows.  

{\mydef {\it A sample-dependent function on the parameter space is a {\sc confidence distribution} for a parameter $\theta$ if 1) For each given sample the function is a distribution function on the parameter space; 2) The function can provide confidence intervals/regions of all levels for $\theta$.} } 
 
\vspace{1.5mm}

A confidence distribution estimator has a similar appeal to a Bayesian posterior in that it is a distribution function carrying much information about the parameter. A confidence distribution however, is a frequentist notion which treats the parameter as a fixed, unknown quantity. It is not a distribution of the parameter; rather, it is a sample-dependent function used to estimate the parameter of interest, including to quantify the uncertainty of the estimation. 

The theoretical foundation for approximate confidence distribution computing relies upon the {\it frequentist coverage property} of confidence distributions. This is the property by which a confidence distribution is able to produce confidence intervals/regions for $\theta$ that contain this true parameter value, $\theta_0$, at any specified frequency.

We hope to demonstrate that the construction of approximate confidence distribution computing as a likelihood-free method provides one of many examples in which confidence distribution theory provides a useful inferential tool for a problem where a statistical method with desirable properties was previously unavailable. 
Furthermore, approximate confidence distribution computing provides a computational method with potential applications extending beyond the scope of Algorithm~\ref{alg:rejABC} and, as will be discussed later, it introduces some flexibility that can greatly decrease computing costs. 

Approximate confidence distribution computing proceeds in the same manner as Algorithm~\ref{alg:rejABC}, but no longer requires a prior assumption on $\theta$; instead, the user is free to select a data-dependent function, $r_{n}(\theta)$, from which potential parameter values will be generated. 
Specifically, the new algorithm proceeds as follows:
\noindent

\begin{algo}
	{(Accept-reject approximate confidence distribution computing) \label{alg:rejACC}} 
	\begin{tabbing}
		\quad 1. Simulate $\theta_{1},\ldots,\theta_{N}\sim r_{n}(\theta)$; \\
		\quad 2. and 3. are identical with steps 2 and 3 of Algorithm \ref{alg:rejABC}.
	\end{tabbing}
\end{algo} 

\noindent
The underlying distribution from which the accepted draws of $\theta$ are simulated is denoted by $Q_{\veps}(\theta\mid s_{\rm obs})$. We refer to $Q_{\veps}(\theta \mid s_{\rm obs})$ as an {\it approximate confidence distribution} and denote the corresponding density by $q_{\veps}(\theta\mid s_{\rm obs})$ as defined by replacing $\pi(\theta)$ in \eqref{ABC_approx_posterior} with $r_{n}(\theta)$:
\begin{align} 
 q_{\veps}(\theta\mid s_{\rm obs})=
\frac{\int_{ \cal S}r_{n}(\theta)f_n(s\mid\theta)K_{\veps}(s - s_{\rm obs})\,ds}{\int_{{\cal P} \times{\cal S}}r_{n}(\theta)f_n(s\mid\theta)K_{\veps}(s - s_{\rm obs})\,ds d\theta}, \label{ACC_approx_posterior}
\end{align}
 In this way, approximate Bayesian computing can be viewed as a special case of approximate confidence distribution computing with $r_{n}(\theta) = \pi(\theta)$. 

From a Bayesian perspective, one may view Algorithm \ref{alg:rejACC} as as an extension permitting the use of Algorithm \ref{alg:rejABC} in the presence of a data-dependent prior. However, there is another natural, frequentist interpretation that views the function $r_{n}(\theta)$ as an initial distribution estimate for $\theta$ and views Algorithm \ref{alg:rejACC} as a method to update this estimate in pursuit of a better-performing distribution estimate. The logic of this frequentist interpretation is analogous to any updating algorithm in point estimation (e.g., say, a Newton-Raphson algorithm or an Expectation-maximization algorithm), which requires an initial estimate and then updates in search for a better-performing estimate. One may ask if the data are thus being `doubly used'. The answer depends on how the initial distribution estimate is chosen. Under some constraints on $r_{n}(\theta)$, Algorithm \ref{alg:rejACC} can guarantee a distribution estimator for $\theta$ that satisfies the frequentist coverage property thus $q_{\veps}(\theta\mid s_{\rm obs})$ can be used to make inferences (e.g., deriving confidence intervals/regions, $p$-values, etc.), although Algorithm \ref{alg:rejACC} may not guarantee `estimation efficiency' (i.e., producing the tightest confidence sets for all levels) unless the summary statistic is sufficient. 

\subsection{Related work}  
Likelihood-free methods such as approximate Bayesian computing have existed for more than 20 years, but research regarding the theoretical properties of these methods is a newly active area, e.g. \cite{Li2016, Frazier2016}. Here we do not attempt to give a full review of all likelihood-free methods, but we acknowledge the existence of alternatives such as indirect inference, e.g. \cite{Creel2013, Gourieroux1993}. 

One of our theoretical results specifies conditions under which Algorithm \ref{alg:rejACC} produces an asymptotically normal confidence distribution. This result, presented in Section \ref{sec:largeSamp}, generalizes the work of \cite{Li2017} on the asymptotic normality of the approximate Bayesian computed posterior. 
%under the Bernstein-von Mises type convergence. 
However, in contrast to these papers, we are not concerned with viewing the result of Algorithm \ref{alg:rejACC} as an approximation to some posterior distribution, rather we focus on the properties and performance of this distribution inherited through its connection to confidence distributions. More importantly, the properties we develop here allow us to conduct  inference while guaranteeing the frequentist coverage property. Additionally, presented separately in Section \ref{sec:thms}, we specify general conditions under which Algorithm \ref{alg:rejACC} can be used to conduct frequentist inference that is beyond the Bernstein-von Mises type convergence, including exact inference that does not rely on any sort of asymptotic (large $n$) assumptions or normally distributed populations.  Aside from the errors of Monte-Carlo approximation and the choice of tolerance level, the exact inference from Algorithm \ref{alg:rejACC} ensures the targeted repetitive coverage rates and type-I errors.

The main goal of the paper is to present the idea that the continued study of likelihood-free methods would benefit from the incorporation of confidence distribution theory. To this end, and for the ease of presentation, we mainly focus on the basic accept-reject version of Algorithm \ref{alg:rejACC}, although we will compare the performance of Algorithm \ref{alg:rejACC} with a typical importance sampling approximate Bayesian computing method and also conclude that much of the existing work in the approximate Bayesian computation literature can also be applied to Algorithm \ref{alg:rejACC} to further improve upon its computational performance as discussed in Sections~\ref{sec:thms} and \ref{sec:discuss}.

\subsection{Notation}
Throughout the paper we will use the following notation. The observed data is $x_{\rm obs} \in \mathscr{X} \subset \mathbb{R}^n$, the summary statistic is a mapping $S_n: \mathscr{X} \rightarrow {\cal S} \subset \mathbb{R}^{d}$ and the observed summary statistic is $s_{\rm obs} = S_n(x_{\rm obs})$. The parameter of interest is $\theta \in \P \subset \mathbb{R}^p$ with $p \leq d \leq n$; i.e. the number of unknown parameters is no greater than the number of summary statistics and dimension of the summary statistic is no greater than the dimension of the data. If some function of $S_n$ is an estimator for $\theta$, we denote this function by $\hat{\theta}_S$. Any function of a particular observation, $s_{\rm obs}$, is therefore an estimate. Let $\theta_0$ represent the fixed, true value of the parameter $\theta$. 
Denote the approximate confidence distribution function by $Q_{\veps}(\theta \mid s_{\rm obs})$, its density function by $q_{\veps}(\theta\mid s_{\rm obs})$, and a random draw from this distribution by $\theta_{\rm ACC}$. Similarly, denote the approximate Bayesian computed posterior distribution by $\Pi_{\veps}(\theta \mid s_{\rm obs})$ and its density function by $\pi_{\veps}(\theta\mid s_{\rm obs})$.
%, and its random draw by $\theta_{\rm ABC}$.
Additionally, %\ST{for a series $z_n$, we use the notation that $z_n \approx a_n$, if there exists constants $m$ and $M$ such that $0<m<|z_n/a_n|<M<1$ as $n \rightarrow \infty$}, and 
for a real function $g(x)$, denote its gradient function at some $x=x_0$ by $D_x\{g(x_0)\}$; for simplicity and when it is clear from context, $x$ is omitted from $D_{x}$. %{\color{red} [Shouldn't this sentence be: ``$x$ is omitted from $D_{x}$"?]}

% [$D_{\theta}$ is not defined, so I guess it is $D_{x}$??]} 
%$D_{\theta}$.   

%%-------------------------------------------------------------------------------------------------------%%

\section{Establishing frequentist guarantees for Algorithm \ref{alg:rejACC}}
\label{sec:thms}
% If the randomness in the Monte-Carlo simulation from $Q_{\veps}$ matches that of the sampling population, then approximate confidence distribution computing can be used to help us answer inference questions with frequentist guarantees on performance.

In this section, we formally % derive this statement and
 establish conditions under which Algorithm \ref{alg:rejACC} can be used to produce confidence regions with guaranteed frequentist coverages at any level.  

To motivate our main theoretical result, we first consider the simple case where we have a scalar parameter, $\theta$, and $\htheta $ is a function that maps the summary statistic into the parameter space $\P$. % = (-\infty, \infty)$.  
Suppose further that the Monte-Carlo copy of $(\theta_{\rm ACC}-\htheta)\mid S_n =s_{\rm obs} $ and the sampling population copy of $(\htheta-\theta)\mid\theta=\theta_0$ have the same distribution: 
\begin{equation}
\label{eq:LR}
(\theta_{\rm ACC}-\htheta)\mid S_n = s_{\rm obs}\sim (\htheta-\theta)\mid\theta=\theta_0.
\end{equation}
Then, we can conduct inference for $\theta$ with a guaranteed frequentist standard of performance. On the left hand side of (\ref{eq:LR}), $\htheta$ is fixed given $s_{\rm obs}$ and the (conditional) probability measure is with respect to $\theta_{\rm ACC}$, meaning the randomness is due to the simulation conducted in Algorithm \ref{alg:rejACC}. Conversely, on the right hand side,  $\htheta$ is a random variable since the data is random for a given parameter $\theta_0$. 
That is, equation (\ref{eq:LR}) states that the `randomness' in $\theta_{\rm ACC}$ from the Monte-Carlo simulation match that in $\htheta$ of the sampling population.
This is very similar to the bootstrap central limit theorem that $n^{1/2}(\theta_{B}-\hat{\theta}_{S})\mid S_n = s_{\rm obs} \sim n^{1/2}(\hat{\theta}_{S}-\theta)\mid\theta=\theta_0$, as $n \to\infty$, where appropriate; cf, \cite{Singh1981} and \cite{Freedman1981}. There, the randomness on the left hand side is from the bootstrap estimator, $\theta_{B}$ given $S_n = s_{\rm obs}$, and the randomness on the right hand side is from the random sample of the sampling population.

Given (\ref{eq:LR}), let $G(t) = \text{pr}(\htheta - \theta \leq t \mid \theta=\theta_0)$. Then $\text{pr}^*(\theta_{\rm ACC} -  \htheta \leq t \mid S_{n} = s_{\rm obs} ) = G(t)$ where $\text{pr}^*(\cdot \mid S_{n} = s_{\rm obs}) $ refers to the probability measure on simulation given $ S_n = s_{\rm obs}$ corresponding to the left hand side of (\ref{eq:LR}). Define $H(t, s_{\rm obs})= \text{pr}^*(2\htheta - \theta_{\rm ACC} \leq t \mid S_{n} = s_{\rm obs} )$, a mapping from $\P \times {\cal S} \rightarrow (0,1)$. %\Minge{The notation $D$ here may be mixed with the $D$ on line 241... perhaps used a new notation?]}  
Conditional on $s_{\rm obs}$,  $H(t, s_{\rm obs})$ is a sample-dependent cumulative distribution function on $\P$; We use the shorthand $H_n(t)$ to denote $H(t, s_{\rm obs})$. The following statement Remark\ref{remk1} holds as proved in the appendix. In the remark, $H_n^{-1}(\alpha)$ is the quantile of  $H_n(\cdot)$, i.e., the solution of  $H_n(t) = \alpha$, and $\theta_{ACC, \alpha}$ is a quantile of $\theta_{\rm ACC}$, defined by $\text{pr}^*(\theta_{\rm ACC} \leq \theta_{ACC,\alpha}\mid S_{n}= s_{\rm obs} ) = \alpha$.

\begin{remark} \label{remk1}
	Under the setup above, $H_n(t)$ is a confidence distribution for $\theta$ and, for any $\alpha \in(0,1)$,  $(-\infty, H_n^{-1}(1-\alpha)] = (-\infty,  2\htheta - \theta_{ACC, \alpha}]$ is an $(1-\alpha)$-level confidence interval of $\theta$. 
% Here $H_n^{-1}(\alpha)$ is the quantile of  $H_n(\cdot)$, i.e., the solution of  $H_n(t) = \alpha$, and $\theta_{ACC, \alpha}$ is a quantile of $\theta_{\rm ACC}$, defined by $\text{pr}^*(\theta_{\rm ACC} \leq \theta_{ACC,\alpha}\mid S_{n}= s_{\rm obs} ) = \alpha$.
\end{remark}
\noindent% 

Now we introduce a key lemma that generalizes the argument above to a multidimensional parameter and a wider range of relationships between $S_n$ and $\theta_{\rm ACC}$. This lemma assumes a relationship between two mappings $V$ and $W: \P \times {\cal S} \rightarrow \mathbb{R}^k$, where $V(\cdot, S_n)$ is a function that acts on the parameter space $\P$, given $S_n = s_{\rm obs}$, and $W(\theta, \cdot)$ is a function that acts on the space of the summary statistic ${\cal S} \subset \mathbb{R}^{d}$, given $\theta=\theta_0$. For example, in the one dimensional argument above, $V(t_1, t_2) = - W(t_1, t_2) = t_1 -\hat{\theta}(t_2)$, where $\hat{\theta}$ is a function of the summary statistic. 
%; however, we may also wish to consider other non-linear mappings. 
Corresponding to (\ref{eq:LR}), we require a matching equation: $V(\theta_{\rm ACC}, S_n)\mid S_n=s_{\rm obs} \sim W( \theta, S_n)\mid\theta = \theta_0 $. Formally, for general mappings $V$ and $W$, we consider Condition~\ref{cond:ACC_interval} below. In the condition, $\delta_\veps \to 0$, as $\veps \to 0$. Here, $\veps$ is the tolerance level for the matching of simulated $s^{(i)}$ and $s_{\rm obs}$ in step 3 of Algorithm \ref{alg:rejACC}, and it may or may not depend on the sample size $n$.
\begin{condition} \label{cond:ACC_interval}
For $\mathfrak{B}$ a Borel set on $\mathbb{R}^k$, % $\veps\rightarrow0$ implies
	\[
	\sup_{A\in\mathfrak{B}}\left\| \text{\rm pr}^*\{V(\theta_{\rm ACC}, S_n) \in A \mid  S_n= s_{\rm obs} \}-  \text{\rm pr}\{W( \theta, S_n) \in A \mid \theta = \theta_0 \}
	\right\| = o_p(\delta_\veps), 
	\] where $\text{\rm pr}^*(\cdot \mid s_{\rm obs}) $ refers to the probability measure on the simulation given $ S_n = s_{\rm obs}$ and $\text{\rm pr}(\cdot \mid \theta_{0})$ is the probability measure on the data before it is observed. 
\end{condition}

% \noindent 
For a given $s_{\rm obs}$ and $\alpha \in (0,1)$,  define a set $A_{1 - \alpha} \subset \mathbb{R}^k$  such that, 
\begin{equation}
\label{eq:A} 
\text{pr}^*\{V(\theta_{\rm ACC}, S_n) \in A_{1 - \alpha} \mid S_n = s_{\rm obs}\} = (1 - \alpha) + o(\delta'), 
\end{equation}  
where $\delta' > 0$ is a pre-selected small positive precision number.
Condition \ref{cond:ACC_interval} implies that 
\begin{equation}
\label{eq:AB}
\Gamma_{1 - \alpha}(s_{\rm obs}) \stackrel{\hbox{\tiny def}} = \{\theta: W(\theta, s_{\rm obs}) \in A_{1 - \alpha} \} \subset {\cal P}
\end{equation}
is a level $(1 - \alpha) 100\%$ confidence region for $\theta_0$.
We summarize this in the following lemma which is proved in the appendix.
Note that in the next lemma, $\delta = \max\{\delta_\veps,\delta'\}$ and there are no requirements on the sufficiency of the summary statistic $S_n$ in the lemma. However, if the selected summary statistic happens to be sufficient, then inference based on the results of Algorithm \ref{alg:rejACC} is equivalent to maximum likelihood inference.  
%\Minge{[With the notation changes, please modify the proof accordingly]}

\begin{lemma}  \label{main1}
	Suppose that there exist mappings $V$ and $W: \P \times {\cal S} \rightarrow \mathbb{R}^k$ such that Condition \ref{cond:ACC_interval} holds.
	Then, $\text{\rm pr}\{\theta \in \Gamma_{1 - \alpha}(S_n) \mid \theta = \theta_0 \} 
	= (1 - \alpha) + o_p(\delta)$. If further Condition \ref{cond:ACC_interval} holds almost surely, then $\text{\rm pr}\{\theta \in \Gamma_{1 - \alpha}(S_n) \mid \theta = \theta_0 \} = (1 - \alpha) + o(\delta)$, almost surely.
\end{lemma}

Often, $\delta'$ in (\ref{eq:A}) is designed to control Monte-Carlo approximation error, thus whether or not Lemma~\ref{main1} is a large sample result depends only on whether or not we require $\veps \to 0$ at a certain rate of the sample size $n$. In the latter part of this section, we will consider a case of Lemma \ref{main1} that is sample-size independent. In this case, aside from the errors of Monte-Carlo approximation and the choice of tolerance level, Algorithm \ref{alg:rejACC} provides an {\it exact} inference that does not rely on large sample asymptotics. Later, in Section \ref{sec:largeSamp}, we extend the large-sample Bernstein-von Mises theory to Algorithm \ref{alg:rejACC}, using a tolerance $\veps$ that depends on $n$.  

Before we move on to verify Condition \ref{cond:ACC_interval} for different cases, we first relate equation (\ref{eq:A}) to $\theta_{\rm ACC}$ samples from $Q_{\veps}(\cdot \mid s_{\rm obs})$. Suppose $\theta_{{\rm ACC}, i}$, $i = 1, \ldots, N$,  are $m$ Monte-Carlo copies of $\theta_{\rm ACC}$. Let ${ v}_i =    V(\theta_{{\rm ACC}, i},  s_{\rm obs})$. The set $A_{1 - \alpha}$ can typically be a $(1-\alpha)100\%$ contour set of $\{{ v}_1, \ldots, { v}_m\}$ satisfying $o(\delta') = o(m^{-1/2})$. For example, we can directly use ${ v}_1, \ldots, { v}_m$ to construct a  $100(1-\alpha)\%$ depth contour as $A_{1 - \alpha} = \{\theta : (1/m)\sum_{i=1}^{m}\mathbb{I}{\{\hat D({ v}_i)< \hat D(\theta)\}} \geq \alpha\}$, where $\hat D(\cdot)$ is an empirical depth function on~$\P$ computed based on the empirical distribution of $\{{ v}_1, \ldots, { v}_m\}$. See, e.g., \cite{Serfling2002} and \cite{Liu1999} for the development of data depth and depth contours in nonparametric multivariate analysis.  In the special case where $k=1$, by defining $\hat q_{\alpha} = { v}_{[m \alpha]}$, the $[m \alpha]$th largest ${ v}_1, \ldots, { v}_m$, a $(1-\alpha)100\%$ confidence region for $\theta_0$ can then be  constructed as $\Gamma_{1 - \alpha}(s_{\rm obs})  = \{ \theta :  \hat q_{\alpha/2}  \leq W(\theta, s_{\rm obs}) \leq \hat q_{1- \alpha/2} \}$ or  $\Gamma_{1 - \alpha}(s_{\rm obs})  = \{ \theta :  W(\theta, s_{\rm obs}) \leq \hat q_{1 - \alpha} \}$.  

We also remark that the existing literature on likelihood-free methods typically relies upon obtaining a ``nearly sufficient'' summary statistic to justify inferential results; see e.g., \cite{Joyce2008}. In this paper however, we explore guaranteed frequentist properties of Algorithm \ref{alg:rejACC} that hold without regard to a ``sufficient enough" summary statistic. However, if the summary statistic happens to be sufficient, then an appropriate choice of the rough initial estimate, $r_{n}(\theta)$, means that inference based on the resulting distribution, $Q_{\veps}(\cdot \mid s_{obs})$, is also efficient.

	To end this section, we explore a special case of Algorithm \ref{alg:rejACC} where the mappings $V$ and $W$ correspond an approximate pivotal statistic. Here, we call 
	a mapping $T = T(\theta, S_{n})$ from $\P \times {\cal S} \to \mathbb{R}^{d}$  an {\it approximate pivot statistic}, if
	\begin{equation}\label{eq:apivot}
	\text{pr}\{T(\theta, S_{n}) \in A \mid \theta = \theta_0\} =  \int_{t \in A}
	g(t) d t \,  \{1 + o(\delta^{''})\}, 
	% \quad \hbox{for $A \subset \mathbb{R}^{d}$,}
	\end{equation}
	where $g(t)$ is a density function that is free of the parameter $\theta$ and $A \subset \mathbb{R}^{d}$ is any Borel set. Also, $\delta^{''}$ is either zero or a small number (tending to zero) that may or may not depend on the sample size $n$. % and the parameter $\btheta$. 
	The usual pivotal cases are special examples of such. Other examples, including that to be discussed in Section~\ref{sec:largeSamp}, involve large sample asymptotics with $\delta^{''}$ is a function of $n$, in particular, $\delta^{''} \to 0$ as $n \to 0$. However, there are also cases where $\delta^{''}$ does not involve the sample size $n$. 
	For example, suppose $S_n | \theta = \lambda \sim \text{Poisson}(\lambda)$. Then, $T( \lambda, S_{n}) = (S_n -  \lambda)/\sqrt{ \lambda}$ is an approximate pivot when $\lambda$ is large. In this case, the  density function is $\phi(t) \{1 + o(\lambda^{-1})\}$, where $\phi(t)$ the density function of the standard normal distribution \cite[]{Cheng1949}.
	
	We have the following theorem for approximate pivot statistics. 
	% Theorem \ref{thm:pivot} . 
	A proof is given in the appendix. 
	\noindent
	
	\begin{theorem} \label{thm:pivot} 
		Suppose $T = T(\theta, S_{n})$ is an approximate pivot statistic that is differentiable with respect to the summary statistic. 
		% Assume $T(\cdot, \cdot)$ is differentiable.
		% with $\frac{\partial}{\partial s}  T(s, \btheta) = a^{-1}(\btheta) \frac{\partial}{\partial \theta}  T(s, \btheta)$ for some function $a(\cdot)$ that is free of $s$.  
		% Suppose, 
		Assume that, for given $t$ and $\btheta$,  $s_{t, \btheta}$ is solution to the equation $t = T(\theta, s)$ and 
		\begin{equation}\label{eq:req}
		\hbox{$\int  r_{n}(\theta) K_{ \veps}\left(s_{t, \theta}-s_{\rm obs}\right)   d \theta = C\{1+ o(\delta^{'}_{\veps})\}$, where $C$ is a constant free of $t$,}
		\end{equation}
		Here,  $r_{n}(\theta)$, $K(\cdot)$, and  $\veps$ are as specified in Algorithm \ref{alg:rejACC}, and $\delta^{'}_{\veps} \to 0$ as 
		$\veps \to 0$. 
		Then, Condition \ref{cond:ACC_interval} holds almost surely, for $V(\theta, S_{n}) = W(\theta,S_{n}) = T( \theta, S_{n})$ and $\delta = \max\{\delta^{''}, \delta^{'}_{\veps}\}$.  
		Furthermore, by Lemma~\ref{main1} and for observed $S_n = s_{\rm obs}$, $\Gamma_{1 - \alpha}(s_{\rm obs})$ defined in~(\ref{eq:AB}) is a level  $(1 - \alpha)100\%$ confidence region with $\text{pr}\{\theta \in \Gamma_{1 - \alpha}(S_n) \mid \theta= \btheta_0 \} = (1 - \alpha) + o(\delta)$, almost surely. 
	\end{theorem} 
	
	Location and scale families contain natural pivot statistics. We verify requirement (\ref{eq:req}) for the location and scale families, which leads to the following corollary. A proof of the corollary is also given in the appendix. 

	\begin{corollary}\label{cor:pivot}
		Assume $\hat{\mu}_S$ and $\hat{\sigma}_S$ are point estimators for location and scale parameters $\mu$ and $\sigma$, respectively.\\
		{ \small Part 1} 
		Suppose $\hat{\mu}_S \sim g_1(\hat{\mu}_S - \mu)$. 
		If $r_{n}(\mu) \propto 1$, then, for any $u$,  
		$$|\text{pr}^*(\mu_{\rm ACC} - \hat{\mu}_S  \leq u\mid \hat{\mu}_{\rm obs})-  \text{pr}(\hat{\mu}_S  - \mu \leq u \mid\mu=\mu_0)| = o(1), \quad\text{almost surely.}$$ 
		{ \small Part 2} 
		Suppose $\hat{\sigma}_S \sim g_2(\hat{\sigma}_S/ \sigma )/\sigma$. 
		If $r_{n}(\sigma) \propto 1 / \sigma$, then, for any $v>0$, 
		$$\left|\text{pr}^*\left(\frac{\sigma_{\rm ACC}}{\hat{\sigma}_S}  \leq v \mid \hat{\sigma}_{\rm obs}\right)-  \text{pr}\left(\frac{\hat{\sigma}_S}{\sigma} \leq v \mid\sigma=\sigma_0 \right) \right| = o(1), \quad\text{almost surely.}$$ 
		{ \small Part 3} 
		Suppose $\hat{\mu}_S \sim g_1\{(\hat{\mu}_S - \mu)/ \sigma \}/\sigma $ and $\hat{\sigma}_S \sim g_2\left( \hat{\sigma}_S/ \sigma\right)/\sigma$ are independent.
		If $r_{n}(\mu, \sigma) \propto 1 / \sigma$, then, for any $u$ and any $v > 0$,  
		\bes 
		&& \left|\text{pr}^*\left(\mu_{\rm ACC} - \hat{\mu}_S  \leq u, \frac{\sigma_{\rm ACC}}{\hat{\sigma}_S}  \leq v \mid \hat{\mu}_{\rm obs},  \hat{\sigma}_{\rm obs}\right)-\right.  \\
		&&\hspace{1cm}\left.\text{pr}\left(\mu_{\rm ACC} - \hat{\mu}_S  \leq u, \frac{\hat{\sigma}_S}{\sigma} \leq v \mid\mu=\mu_0, \sigma=\sigma_0 \right)\right|
		= o(1), \quad\text{almost surely.}
		\ees 
		Furthermore, we may derive
		$ H_1(\hat{\mu}_S, x) = 1 - \int_{-\infty}^{\hat{\mu}_S-x}g_1(w)\,dw $, a confidence distribution for $\mu$ induced by $(\hat{\mu}_S-\mu)$ given $\mu=\mu_0$, or
		$H_2(\hat{\sigma}_S^2,x) = 1 - \int_{0}^{\hat{\sigma}_S^2/ x} g_2(w)dw$, 
		a confidence distribution for $\sigma^2$ induced by $\hat{\sigma}_S^2/ \sigma^2$ given $\sigma = \sigma_0$.
	\end{corollary}

	Note that Theorem \ref{thm:pivot} and Corollary \ref{cor:pivot} cover some finite sample examples that do not require $n\rightarrow\infty$, one of which is illustrated in Figure~\ref{fig:ABC}. Specifically, Corollary \ref{cor:pivot} Part 1 suggests that the ABC posteriors obtained in the Cauchy example in Figure~\ref{fig:ABC}, using either the sample mean or sample median as the summary statistic, are both confidence distributions. Thus, they both are `distribution estimators' that can be utilized to make inference. Both are not efficient, and the one by the sample median is more efficient than the one by the sample mean (in terms of having shorted confidence intervals or a higher power level-$\alpha$ test). This development represents a departure from the typical asymptotic arguments and permits the use of Algorithm \ref{alg:rejACC} in forming confidence intervals/regions with guaranteed frequentist coverages even when $n$ is finite. 
	% Specifically, when we utilize an exact pivot, the only source of approximation in the inference resulting from Algorithm \ref{alg:rejACC} is due to the computational requirements on $\veps$. 

The next section considers the case in which the tolerance level $\veps$ does depend on the sample size $n$. We will now denote $\veps$ by $\veps_n$ and study the large sample performance of the proposed approximate confidence distribution computing method.

%-------------------------------------------------------------------------------------------%
\section{Frequentist Coverage of Algorithm \ref{alg:rejACC} for Large Samples} 
\label{sec:largeSamp}
\subsection{Bernstein-von Mises theorem for Algorithm \ref{alg:rejACC}}
For Algorithm \ref{alg:rejABC}, Condition \ref{cond:ACC_interval} holds as $n\rightarrow\infty$ by the Bernstein-von Mises type convergence of $\pi_{\veps}(\theta\mid s_{\rm obs})$ \cite[]{Li2016} and selecting $\veps_n$ decreasing to zero. Roughly speaking, the distribution of a properly scaled draw from $\Pi_{\veps}(\theta \mid s_{\rm obs})$ and the distribution of the corresponding expectation (before the data is observed) are asymptotically the same. Therefore, the development in Section~\ref{sec:thms}, a confidence region with asymptotically correct coverage can be constructed using a sample from Algorithm \ref{alg:rejABC}. 

Here we show that Condition \ref{cond:ACC_interval} also holds for the more general Algorithm \ref{alg:rejACC} where $r_{n}(\theta)$ may depend upon the data. The results are based on the same set of conditions as those in \cite{Li2016}. The key condition is a central limit theorem of the summary statistic: for all $\theta$ in a neighborhood of $\theta_0$,
\[
a_{n}\{\bS_{n}- \eta(\btheta)\}\rightarrow N\{0,A(\btheta)\},
\]
in distribution as $n\rightarrow\infty$, together with requirement on the identifiability of $\theta_0$ through $\eta(\theta)$ and regulatory requirements of $A(\theta)$. This condition is denoted by Condition \ref{sum_conv} in the supplementary materials. For convenience, the set of conditions in \cite{Li2016} is given in the supplementary materials. Additionally, some regulatory conditions for $r_{n}(\theta)$ are listed below. 

\begin{condition} \label{par_true}
There exists some $\delta_0 > 0$ such that $\mathcal{P}_0 = \{\theta: \| \theta-\theta_0\|  < \delta_0  \} \subset \mathcal{P},$ $r_{n}(\theta) \in C^2(\mathcal{P}_0)$, and $r_{n}(\theta_0)>0$.
\end{condition}

\begin{condition} \label{initial_upper}
There exists a sequence $\{\tau_{n}\}$ and $\delta>0$, such that $\tau_{n}=o(a_n)$ and $\sup_{\btheta\in {\cal P}_{0}}\tau_{n}^{-p}r_{n}(\btheta)=O_{p}(1)$.
\end{condition}

\begin{condition} \label{initial_lower}
%It holds that %$r_{n}(\btheta)\in C^{1}(\mathcal{P}_{0})$ %%this is redundant 
%$\tau_{n}^{-p}r_{n}(\btheta_{0})=\Theta_{p}(1)$.
There exists constants $m$, $M$ such that $0 < m <\mid \tau_{n}^{-p}r_{n}(\btheta_{0})\mid < M < \infty$.
\end{condition}

\begin{condition} \label{initial_gradient}
It holds that $\sup_{\btheta\in\mathbb{R}^{p}}\tau_{n}^{-1} D\{\tau_{n}^{-p}r_{n}(\btheta)\}=O_{p}(1)$.
\end{condition}

%In the above, the notation for a series $x_n = \Theta(a_n)$ means that there exists some constants $m$ and $M$ such that $0 < m <\mid x_n/a_n \mid < M < \infty$. 
Condition \ref{initial_upper} and \ref{initial_lower} above essentially requires $r_{n}(\theta)$ to be more dispersed than the {\it s-}likelihood for within a compact set by requiring that $r_{n}(\theta)$ converges to a point mass more slowly than $f_{n}(\theta \mid s_{\rm obs})$. Condition \ref{initial_gradient} requires the gradient of the standardized $r_{n}(\theta)$ to converge with rate $\tau_n$. These are relatively weak conditions and can be satisfied by, e.g., $r_{n}(\btheta)$ satisfying local asymptotic normality. We have the following theorem with the proof provided in the appendix.
 Note that, in the theorem, $\theta_{\veps}(s_{\rm obs})$  is an estimate for $\theta$, whereas $\theta_{\veps}(S_{n})$ is an estimator; when clear, we shorten the notation of both to $\theta_{\veps}$.

\begin{theorem} \label{thm:ACC_limit_small_bandwidth}
	Assume $r_{n}(\btheta)$ satisfies Condition \ref{par_true}--\ref{initial_gradient} and \ref{sum_conv}--\ref{cond:likelihood_moments} in the supplementary material. If $\veps_{n}=o(a_{n}^{-1})$ as $n\rightarrow\infty$, then Condition \ref{cond:ACC_interval} is satisfied with $V(\btheta_{\rm ACC},\bsob)=a_{n}\{\btheta_{\rm ACC}-\btheta_{\veps}(s_{\rm obs})\}$ and $W(\btheta_0,S_{n})=a_{n}\{\btheta_{\veps}(S_{n})-\btheta_{0}\}$, where % $\btheta_{\veps} = 
$\btheta_{\veps}(s) =  \int \btheta  \, d Q_{\veps}(\theta\mid s)$.
\end{theorem}

Theorem \ref{thm:ACC_limit_small_bandwidth} says when $\veps_{n}=o(a_{n}^{-1})$, the coverage of $\Gamma_{1-\alpha}(s_{\rm obs})$ is asymptotically correct as $m_{\rm ACC}\rightarrow\infty$ and $n\rightarrow\infty$, where $m_{\rm ACC}$ is the number of accepted particles in Algorithm 2. In practice, $\theta_{\veps}(s_{\rm obs})$, needed for constructing $\Gamma_{1-\alpha}$, does not have a closed form in most cases, and is estimated by the sample of $\theta_{\rm ACC}$. 

In Theorem \ref{thm:ACC_limit_small_bandwidth}, Condition \ref{cond:ACC_interval} is implied by the following convergence results, 
\begin{equation}
\sup_{A\in\mathfrak{B}^{p}}\left| \int_{\{\btheta: \, a_{n}(\btheta-\btheta_{\veps})\in A\}} d Q_{\veps}(\btheta \mid  \bS_n = \bs_{{\rm obs}})-\int_{A}N\{t;0,I(\btheta_{0})^{-1}\}\,dt\right|\rightarrow0, \label{thm2_uncertainty}
\end{equation}
in probability, and 
\begin{equation}
a_{n}(\btheta_{\veps}-\btheta_{0}) \rightarrow N\{0,I(\btheta_{0})^{-1}\}, \label{thm2_ptestimate}
\end{equation}
in distribution, as $n\rightarrow\infty$, where $I(\theta)=D\eta(\theta)^TA^{-1}(\theta)D\eta(\theta)$. These results generalize the limit distributions of $\Pi_{\veps}$ in \cite{Li2017} for the case of $\veps_n=o(a_n^{-1})$, since the prior distribution $\pi(\theta)$ satisfies Condition \ref{par_true}--\ref{initial_gradient}. We show that, in the sense of large-sample behavior, inference based on $Q_{\veps}$ is validated whether or not information from the data is used in constructing $r_{n}(\theta)$.

\subsection{Comparison between Algorithm \ref{alg:rejABC} and Algorithm \ref{alg:rejACC}}

Since $\Pi_{\veps}$ and $Q_{\veps}$ share the same limit distributions according to \eqref{thm2_uncertainty} and \eqref{thm2_ptestimate}, when the same tolerance level is used, confidence regions $\Gamma_{1-\alpha}(\bsob)$ constructed using the sample from $\Pi_{\veps}$ and $Q_{\veps}$ have the same asymptotic efficiency. Therefore it is computationally more efficient to use Algorithm \ref{alg:rejACC} with $r_{n}(\theta)$ depending on data, since any $r_{n}(\theta)$ with $\tau_n\rightarrow\infty$ is closer to the output distribution than $\pi(\theta)$ thus providing a higher acceptance probability for the same $\veps$. 

When $r_{n}(\theta)$ is available, an alternative to Algorithm \ref{alg:rejABC} is its importance sampling variant which proposes from $r_{n}(\theta)$ \cite[]{Fearnhead2012}, as specified in the following.
\begin{algo}{(Importance sampling approximate Bayesian computing)} \label{alg:ISABC}
	\begin{tabbing}
		\quad 1. Simulate $\btheta_{1},\ldots,\btheta_{N}\sim r_{n}(\btheta)$. \\
		\quad 2. For each $i=1,\ldots,N$, simulate $\bx^{(i)}=\{x_{1}^{(i)},\ldots,x_{n}^{(i)}\}$ from $M_{\btheta}$.\\ %\sim f_{n}(x\mid \btheta_{i})$;\\
		\quad 3. For each $i=1,\ldots,N$, accept $\btheta_{i}$ with probability $K_{\veps}(s^{(i)}-s_{\rm obs})$, where $s_{\rm obs}= $ \\ \qquad $S_{n}(\bx_{\rm obs})$ and $ s^{(i)}=S_{n}(\bx^{(i)})$, and assign importance weights $w(\theta_i)=\pi(\theta_i)/r_{n}(\theta_i)$.
	\end{tabbing}
\end{algo}  
Though Algorithm~\ref{alg:ISABC} is an improvement over Algorithm~\ref{alg:rejABC}, Algorithm \ref{alg:rejACC} still has a computational advantage over Algorithm \ref{alg:ISABC}, because $w(\theta)$ is unbounded as $n\rightarrow\infty$ while the sample weights in Algorithm \ref{alg:rejACC} are unity. 
\cite{Li2016} mention that certain techniques can be applied to control the skewed importance weight in Algorithm \ref{alg:ISABC}, but Algorithm \ref{alg:rejACC} does not have the same issue and therefore does not require such controls. 

\cite{Li2016} point out in Algorithms~\ref{alg:rejABC} and \ref{alg:ISABC}, that although using $\veps_n=o(a_n^{-1})$ gives valid inference, this leads to the degeneracy of Monte Carlo efficiency as $n\rightarrow\infty$, since the acceptance probability of any proposal distribution degenerates to zero for such a small tolerance level. This means that if the dataset is informative, most of the simulated datasets in Algorithms~\ref{alg:rejABC} and \ref{alg:ISABC},
 will be wasted. If $\veps_n$ is outside this regime, \cite{Li2016} show that $\Pi_{\veps}$ over-inflates the target posterior uncertainty and is not calibrated, i.e its uncertainty can not correctly quantify the uncertainty of the target posterior mean.
 % \Minge{$\theta_{\veps}$ [ST- this is not defined, is this even necessary?]},
 %therefore Condition~\ref{cond:ACC_interval} does not hold. 
 A similar phenomena occurs in Algorithm~\ref{alg:rejACC} when too large $\veps_n$ is used. Instead of giving a formal statement, we illustrate this in the following basic Gaussian example.

\begin{example}
	Consider a univariate normal model with mean $\theta$ and unit variance, and observations that are independent identically distributed from the model with $\theta=\theta_0$. Assume a standard normal density for the prior density of $\theta$, and use the normal density with mean $\mu_n$ and variance $b_n^{-2}$ for $r_{n}(\theta)$, where $\mu_n$ and $b_n$ are some sequences satisfying $b_n(\mu_n-\theta_0)=O(1)$ and $b_n=o(\sqrt{n})$ as $n\rightarrow\infty$. The choice of $\mu_n$ and $b_n$ makes $r_{n}(\theta)$ a reasonable proposal density, since it covers the true parameter $\theta_0$ and is more dispersed than the 
	{\it s-}likelihood where the sample mean is the summary statistic in both Algorithm \ref{alg:rejABC} and \ref{alg:rejACC}. The Gaussian kernel with variance $\veps_n^2$ is used for the acceptance/rejection.
	
	For this model, limit distributions of $V(\theta_{\rm ACC},s_{\rm obs})$ and $W(\theta_0,S_{n})$ in Theorem \ref{thm:ACC_limit_small_bandwidth} for different regimes of $\veps$ can be obtained analytically, since $q_{\veps}(\theta \mid s_{\rm obs})$ has the closed form $N(\theta;\theta_{\veps},\sigma_{\veps}^{2})$ where 
	\[
	\theta_{\veps}=\frac{\bsob+b_{n}^{2}(1/n+\veps^{2})\mu_{n}}{1+b_{n}^{2}(1/n+\veps^{2})},\ \sigma_{\veps}^{2}=\frac{1/n+\veps^{2}}{1+b_{n}^{2}(1/n+\veps^{2})}.
	\]
	In order for Condition \ref{cond:ACC_interval} to hold, $V(\theta_{\rm ACC},s_{\rm obs})$, which has the density $N(\cdot\,; 0,n\sigma_{\veps}^2)$,  and $W(\theta_0,S_{n})$, which is equal to $\sqrt{n}(\theta_{\veps}-\theta_0)$, should have the same asymptotic distributions. 
	
	By decomposing $W(\theta_0,S_{n})$ into $\Delta_{1}\sqrt{n}(S_{n}-\theta_{0})+\Delta_{2}b_{n}(\mu_{n}-\theta_{0})$ where 
	\[
	\Delta_{1}=\frac{1}{1+b_{n}^{2}(1/n+\veps^{2})},\ \Delta_{2}=\frac{\sqrt{n}b_{n}(1/n+\veps^{2})}{1+b_{n}^{2}(1/n+\veps^{2})},
	\]
	it can be seen that the expectation of $W(\theta_0,S_{n})$ is $o(1)$ only when $\veps_n=o(b_{n}^{-1/2}n^{-1/4})$. On the other hand, the variance of $W(\theta_0,S_{n})$ and $n\sigma_{\veps}^2$ having the same limit requires $n\sigma_{\veps}^2-\Delta_{1}^2=o(1)$ which holds only when $\veps_n=o(n^{-1/2})$ or $\veps_n^{-1}=o(b_n^2n^{-1/2})$. Because $b_n=o(\sqrt{n})$, both $\veps_n=o(b_{n}^{-1/2}n^{-1/4})$ and $\veps_n^{-1}=o(b_n^2n^{-1/2})$ can not hold simultaneously. Therefore Condition \ref{cond:ACC_interval} is satisfied only when $\veps_n=o(n^{-1/2})$. 
\end{example}

One remedy to reduce the overinflated uncertainty in $\Pi_{\veps}(\theta\mid\bsob)$ from Algorithms~\ref{alg:rejABC} and \ref{alg:ISABC}
is to post-process its sample by the regression adjustment \cite[]{beaumont2002}. Likewise, this adjustment can be applied to Algorithm \ref{alg:rejACC}. In the next subsection, we compare these regression adjusted approximate computing methods. 
%\Minge{Add comment about how this similarly applies to Algorithm \ref{alg:rejACC}] }

\subsection{Comparison between Algorithm \ref{alg:rejABC} and Algorithm \ref{alg:rejACC} with regression adjustment} \label{regression_discussion}

For Algorithms~\ref{alg:rejABC} and \ref{alg:ISABC}, it is known that the distribution of the regression adjusted sample is able to correctly quantify the posterior uncertainty and yield an accurate point estimate with $\veps_n$ decaying in the rate of $o(a_n^{-3/5})$, which is slower than $o(a_n^{-1/2})$ \cite[]{Li2017}. Here, we suggest applying the same regression adjustment to Algorithm \ref{alg:rejACC} to produce valid inference on the sample of Algorithm \ref{alg:rejACC} with a larger $\veps_n$. 

Let $q_{\veps}(\theta, s)$ be the joint density of accepted $\theta$ and its associated summary statistic in Algorithm \ref{alg:rejACC}, i.e. 
\begin{align} 
&q_{\veps}(\theta, s) =  \frac{r_{n}(\theta)f_{n}(s\mid\theta)K_{\veps}(s - s_{\rm obs})}{\int_{\mathbb{R}^p\times\mathbb{R}^{d}}r_{n}(\theta)
	f_{n}(s\mid\theta)K_{\veps}(s - s_{\rm obs})\,d\theta d s}, \label{ACC_density}
\end{align}
where $\theta\in\mathbb{R}^p$ and $s\in\mathbb{R}^d$.
Denote a sample from $q_{\veps}(\theta,s)$ by $\{(\btheta_{i},\bs^{(i)})\}_{i=1,\ldots,N}$. A new sample can be obtained as $\{\btheta_{i}-\hat{\beta}_{\veps}(\bs^{(i)}-\bs_{{\rm obs}})\}_{i=1,\ldots,N}$
where $\hat{\beta}_{\veps}$ is the least square estimate of
the coefficient matrix in the linear model
\[
\btheta_{i}=\alpha+\beta(\bs^{(i)}-\bs_{{\rm obs}})+e_{i},\ i=1,\ldots,N,
\]
where $e_{i}$ are independent identically distributed errors, $\alpha\in\mathbb{R}^p$ and $\beta\in\mathbb{R}^{p\times d}$. Let $\btheta_{\rm ACC}^{*}=\btheta-\beta_{\veps}(\bs-\bs_{{\rm obs}})$, where $\beta_{\veps}$ is from the minimizer %\Minge{[Comment: In the paper $\btheta$ and $\bs$ are vectors. I think we should reflect that for the notations in the regression and its coefficients.]}%
\[
(\alpha_{\veps},\beta_{\veps})={\rm argmin}_{\alpha\in\mathbb{R}^p,\beta\in\mathbb{R}^{d\times p}}E_{\veps}\left\{\|\btheta-\alpha-\beta(\bs-\bs_{{\rm obs}})\|^{2}\mid \bs_{{\rm obs}}\right\}
\]
for expectation under the joint distribution $q_{\veps}(\btheta,\bs)$. The new sample can be seen as a draw from the distribution of $\btheta_{\rm ACC}^{*}$
where $(\btheta, \bs)\sim q_{\veps}(\btheta,\bs)$, but with $\beta_{\veps}$ replaced by its estimator. Let $\btheta_{\veps}^{*}$
be the expectation of $\btheta_{\rm ACC}^{*}$.
%  \Minge{Denote by $Q_{\veps}^*(\cdot |  \bS_n = \bs_{{\rm obs}})$ the distribution of $\btheta_{\rm ACC}^{*}$ given $\bS_n = \bs_{{\rm obs}}$.}

The following theorem states that the regression adjusted  $Q_{\veps}$ has the same favored property as the adjusted $\Pi_{\veps}$. Here, the regression adjusted $Q_{\veps}$, say $Q_{\veps}^*(\cdot |  \bS_n = \bs_{{\rm obs}})$, is the distribution of $\btheta_{\rm ACC}^{*}$ given $\bS_n = \bs_{{\rm obs}}$.
\begin{theorem}
	\label{thm:ACC_limit_large_bandwidth}
	Assume the conditions of Theorem \ref{thm:ACC_limit_small_bandwidth} and Condition \ref{cond:likelihood_moments} of the supplementary materials.
	If $\veps_{n}=o(a_{n}^{-3/5})$ as $n\rightarrow\infty$, Condition \ref{cond:ACC_interval} is satisfied with $V(\btheta_{\rm ACC}^{*},s_{\rm obs})=a_{n}(\btheta_{\rm ACC}^{*}-\btheta_{\veps}^*)$ and $W(\btheta_0,S_{n})=a_{n}(\btheta_{\veps}^*-\btheta_{0})$. 
\end{theorem}

In the above, Condition \ref{cond:ACC_interval} is implied by the following convergence results which generalize the results in \cite{Li2017},
\[
\sup_{A\in\mathfrak{B}^{p}}\left| \int_{\{\btheta: \, a_{n}(\btheta -\btheta_{\veps}^{*}) \in A\}} d Q_{\veps}^*(\btheta \mid  \bS_n = \bs_{{\rm obs}}) -\int_{A}N\{t;0,I(\btheta_{0})^{-1}\}\,dt\right|\rightarrow0,
\]
in probability, and 
\[
a_{n}(\btheta_{\veps}^{*}-\btheta_{0}) \rightarrow N\{0,I(\btheta_{0})^{-1}\}, 
\]
in distribution, as $n\rightarrow\infty$. The limit distributions above are the same as those in \eqref{thm2_uncertainty} and \eqref{thm2_ptestimate}, therefore $\Gamma_{1-\alpha}(\bsob)$ constructed using $\btheta_{\rm ACC}^{*}$ can achieve the same efficiency as those using $\btheta_{\rm ACC}$ while permitting much larger tolerance levels. Asymptotically, inference based on the regression adjusted $Q_{\veps}^*$ is not affected by an $r_{n}(\theta)$ that depends on the data,  again illustrating the computational advantage of Algorithm \ref{alg:rejACC}.

\subsection{Guidelines for Selecting $\r$ in Algorithm \ref{alg:rejACC}} \label{sec:rn}
The generality of approximate confidence distribution computing is that it can produce justifiable inferential results with weak conditions on a possibly data-dependent function $r_{n}(\theta)$. % from which we can generate initial $\btheta$ values. 
In general, one should be careful in choosing $r_{n}(\theta)$ to ensure its growth with respect to the sample size is slower than the growth of the {\it s-}likelihood, according to Condition \ref{initial_upper}. A generic algorithm to construct $r_{n}(\theta)$ based on sub-setting the data is proposed below. Assume that a point estimator $\widehat{\theta}(z)$ of $\theta$ can be computed for a dataset $z$ of any size. 
%\Minge{[Comment: Do we have any particular reason to use $z$ to denote the dataset? Can we still use $x$? If we like to use $z$, state its connection to $x$.]}

\begin{algo}{(Minibatch scheme)} \label{alg:MB_rn}
	\begin{tabbing}
		\qquad \enspace 1. Choose $k$ subsets of the observations, each with size $n^{\nu}$ for some $0<\nu<1$. \\
		\qquad \enspace 2. For each subset $z_{i}$ of $x_{\rm obs}$, compute the point estimate $\widehat{\theta}_{i}=\widehat{\theta}(z_{i})$, for $i=1,\ldots,k$. \\
		\qquad \enspace 3. Let 
		$r_{n}(\theta)=(1/ kh) \sum_{i=1}^{k}K\left\{h^{-1}\|\theta-\widehat{\theta}_{i}\|\right\},$
		where $h>0$ is the bandwidth of the\\ \qquad \enspace \quad 
		kernel density  estimate using $\{\widehat{\theta}_{1},\ldots,\widehat{\theta}_{k}\}$ and kernel function $K$.
	\end{tabbing}
\end{algo}

By choosing $\nu<3/5$, we ensure that Conditions \ref{initial_upper}--\ref{initial_gradient} are met. Furthermore, if $\widehat{\theta}(z)$ converges with a rate not faster than that of the summary statistic, then the tolerance level, $\veps_{n}$, selected by accepting a reasonable proportion of simulations is sufficiently small, provided the rate of $S_n$ is a power function of $n$. Based on our experience, if $n$ is large one may simply choose $\nu =1/2$ to partition the data. For small $n$, say $n<100$, it is better to select $\nu >1/2$ and overlap the subsets so that each subset contains a reasonable number of observations. 

The choice of $\widehat{\theta}$ does not have to be very accurate, since it is only used to construct the initial  estimate, $r_{n}(\theta)$. For problems of intractable likelihoods, possible choices of $\widehat{\theta}$ include the point maximizing an easy-to-obtain approximate likelihood or the point minimizing the average distance between the simulated $s$ and $\bsob$ \cite[]{meeds2015optimization}. However, a poor choice, for instance, a $\widehat{\theta}$ with a large bias, might cause bias in the inference if the mass of $Q_{\veps}(\theta\mid\bsob)$ is not well covered by the simulated parameter values. For a subset, $z_i$, of the data, $x_{\rm obs}$, we suggest choosing the point estimate to be the
s-likelihood-based expectation over the subset, 
i.e. $E\{\theta\mid S_{n}(z_i)\} \propto \int \theta f_n\{S_{n}(z_i) \mid \theta\}d\theta$. 
%\Minge{[comment: formally define $E$ here]}. 
This choice of $\widehat{\theta}$ has two benefits. First, when the summary statistic satisfies Condition \ref{sum_conv}, $E\{\theta\mid S_{n}(z_i)\}$ is asymptotically unbiased. Second, $E\{\theta\mid S_{n}(z_i)\}$ converges with the same rate as $S_{n}$, which is desirable as discussed above.  

For each subset $z_i$ of $x_{\rm obs}$, $E\{\theta\mid S_n(z_i)\}$ can be approximated using the population Monte Carlo variant of Algorithm \ref{alg:rejABC} \cite[]{Beaumont2009,del2012adaptive}. This variant extends the importance sampling step of Algorithm \ref{alg:ISABC} to a sequence of sampling importance resampling operations, in order to iteratively update the approximate posterior distribution starting from the prior distribution. 
%\Minge{[Comment: perhaps add one or two sentence more here? I do not know how this MC approximation is done just based on the description here.]} 
For an initial choice of $\widehat{\theta}$, say $\widehat{\bar \theta}$, let $\bar r_{n}(\theta)$ be the proposal distribution constructed by Algorithm \ref{alg:MB_rn} together with $\widehat{\bar \theta}$. Here the user can now propose from $\bar r_{n}(\theta)$ in the first iteration of the algorithm rather than proposing from the prior distribution, helping to reduce the associated computational cost. This approximation is straightforward to execute in parallel for multiple subsets and can be applied to Algorithm \ref{alg:rejACC} as well. We call this scheme the \textit{refined-minibatch scheme}, since it updates the $r_{n}(\theta)$ obtained from the minibatch scheme (i.e. Algorithm \ref{alg:MB_rn}) by improving the quality of $\widehat{\theta}$. From our experience, the additional computational cost of the refined version is relatively small compared to the other parts of Algorithm \ref{alg:rejABC} and \ref{alg:rejACC} because a small particle size and several iterations are usually enough to achieve convergence of the population Monte Carlo algorithm with the proposed techniques. A full study on the choice of $\widehat{\btheta}$ is beyond the scope of this paper. 

\begin{remark}\label{remk2}
There is a trade-off in Algorithm \ref{alg:rejACC} between faster computations and guaranteed frequentist inference. When the growth of $r_{n}(\theta)$ is at a similar rate as the {\it s-}likelihood while the sample size $n \to \infty$, the computing time may be reduced but Algorithm \ref{alg:rejACC} may also risk violating Conditions \ref{initial_upper}--\ref{initial_gradient}. If these assumptions are violated, the resulting simulations do not necessarily form a confidence distribution and consequently, inference based on Algorithm \ref{alg:rejACC} may not be valid in terms of producing confidence sets with guaranteed coverage. However, if Conditions \ref{initial_upper}--\ref{sum_conv} do hold and the observed data is large enough, Theorem \ref{thm:ACC_limit_small_bandwidth} shows that regardless of the choice of $r_{n}(\theta)$, Algorithm \ref{alg:rejACC} always produces the same confidence distribution. 
\end{remark}

%--------------------------------------------------------------------------------------------%

\section{Empirical Examples} \label{sec:ex}
\subsection{Cauchy data} \label{sec:Cauchy}
In Figure \ref{fig:ABC} we saw how the lack of a sufficient statistic could drastically change inference resulting from approximate Bayesian computing. In particular, we saw that when applying Algorithm \ref{alg:rejABC}, the approximate Bayesian computational posterior is quite different from the target posterior, $p(\theta \mid x) \propto \prod_{i=1}^{n} 1/ [ 1 + \{(x_i - \theta)/0.55 \}^2 ]$ for both $S_n = \bar{x}$ and $S_n = Median(x)$. 

Now, as a continuation of the example in Figure~\ref{fig:ABC}, suppose we observe random data, $(x_1,\dots,x_n)$, from a $Cauchy(\theta, \tau)$ distribution with sample size $n=400$. Suppose the (unknown) data-generating parameter value is $(\theta_0,\tau_0)=(10,0.55)$. 
Using the coverage of the data-generating parameter as our metric, we compare the performance of the resulting $95\%$ confidence intervals/regions from Algorithm \ref{alg:rejACC}, denoted r-ACC, to the credible intervals/regions from Algorithm \ref{alg:ISABC}, denoted IS-ABC. For the credible intervals/regions of IS-ABC, we also compute the corresponding Bayesian coverage probabilities. 

The reason we choose to compare the rejection sampling version of approximate confidence distribution computing to the importance sampling version of approximate Bayesian computing is illustrated in Table \ref{EX1_rej}. In this table,  we fix $\veps_{n}$ (thus varying the number of retained $\theta$ values, $N$, in each run). Both Algorithm \ref{alg:rejABC} (r-ABC) and Algorithm \ref{alg:rejACC} (r-ACC) suffer from over-coverage, but the acceptance rates for r-ACC are much better and are comparable to Algorithm \ref{alg:ISABC} (IS-ABC) without the regression adjustment.
\begin{table}
	\caption{Comparison of r-ABC, IS-ABC, and r-ACC without the regression adjustment for inference on $\theta$ using the median as the summary statistic and assuming a flat prior on $\theta$. We fix $\veps_{n}$ and compare the median acceptance proportions of each algorithm using a Monte Carlo sample size of $10^6$. Coverage is computed over 300 runs. IS-ABC and r-ACC perform similarly.}{
		\begin{tabular}{ccccccc}
			    &\multicolumn{2}{c}{r-ABC}       &\multicolumn{2}{c}{IS-ABC} &\multicolumn{2}{c}{r-ACC}    \\
			$\veps_{n}$	&Coverage &Acceptance &Coverage &Acceptance &Coverage &Acceptance\\
			& &proportion  & &proportion & &proportion\\
			$0.1 $       & $ 0.973$ & $ 0.0001$ & $1$ & $ 0.001$& $ 1$ & $ 0.001$ \\
			$0.01 $     & $ 1$ & $ 0.001$ & $ 1$& $ 0.008$ & $ 1$ & $ 0.008$  \\
			$0.001$      & $ 1$ & $ 0.008$ & $1$& $ 0.079$& $ 1$ & $0.079 $  \\
	\end{tabular} }
	\label{EX1_rej}
\end{table} 

For reference, all experiment settings mentioned below are summarized in Table \ref{EX1_setting}. The prior distribution used in each of the Bayesian methods is the Jeffrey's prior for the location-scale family. We also compute the median width of the intervals from each experiment. Coverage proportions closer to the $95\%$ nominal level and having smaller width, indicative of higher efficiency, are preferred. Only those results using the regression adjusted sample are reported because in all cases, intervals constructed by the unadjusted samples are much wider and over-cover the true parameter values for almost all acceptance proportions as demonstrated in Table~\ref{EX1_rej}. This aligns with the discussion in Section \ref{regression_discussion}. Under each of these settings, $r_{n}(\theta)$ for r-ACC is constructed using Algorithm \ref{alg:MB_rn} with $\nu =1/2$.% and the sample mean as the point estimate. 

\begin{table}
	\caption{Experiment settings of Example 1. Improper priors are considered for (i)--(iv), and $t_{4}(\mu,\sigma)$ denotes the Student's t density with degree of freedom four, location $\mu$ and scale $\sigma$.  The summary statistic ${\rm MAD}(x)$ is the sample median absolute deviation.}{
		\begin{tabular}{cccc}
			& Unknown parameter & Prior density & Summary statistic \\
			(i) & $\theta$ & $1$ & $\mbox{median}(x)$\\
			(ii) & $\tau$ & $\tau^{-1}I_{\tau>0}$ & $\mbox{MAD}(x)$\\
			(iii) & $(\theta,\tau)$ & $\tau^{-1}I_{\tau>0}$ & $\{\mbox{median}(x),\mbox{MAD}(x)\}$\\
			(iv) & $\theta$ & $1$ & $\overline{x}$\\
			(v) & $\theta$ & $t_{4}(\theta_{0},1)$ & $\overline{x}$\\
			(vi) & $\theta$ & $t_{4}(\theta_{0},3)$ & $\overline{x}$\\
			(vii) & $\theta$ & $t_{4}(\theta_{0}+3,3)$ & $\overline{x}$
	\end{tabular}}
	\label{EX1_setting}
\end{table}

First, we consider inference for one or both of the unknown parameters in settings (i)-(iii) in Table \ref{EX1_setting}. We choose the summary statistics as the sample median and sample median absolute deviation for the location and scale parameters, respectively. These summary statistics are asymptotically normal and unbiased, and satisfy Condition \ref{sum_conv}; thus Theorem \ref{thm:ACC_limit_large_bandwidth} guarantees at least nominal coverage for the intervals/regions of r-ACC, as observed in Table \ref{EX1_results}.A. 

For the first three inference problems in Table \ref{EX1_results}, approximate confidence distribution computing (i.e. Algorithm \ref{alg:rejACC}) and importance sampling approximate Bayesian computing (i.e. Algorithm \ref{alg:ISABC}) perform very similarly. This is not surprising, since the data size is large enough that the asymptotic behaviors of all estimates are similar. As discussed in Section \ref{regression_discussion}, $Q_{\veps}(\cdot \mid s_{obs})$ and $\Pi_{\veps}(\cdot \mid s_{obs})$ share the same limiting normal distribution and thus the credible intervals/region of the approximate posterior from IS-ABC are similar to the confidence intervals/region of the confidence distribution from r-ACC.

For the last four inference problems in Table~\ref{EX1_setting}, we wish to conduct inference on the location parameter only but we choose a less informative summary statistic, the sample mean. This summary statistic follows a Cauchy$(\theta,\tau)$ distribution and thus does not satisfy Condition \ref{sum_conv}. However, by  Theorem \ref{thm:pivot}, we are still able to produce confidence intervals with nominal coverage using r-ACC. 
In Table \ref{EX1_results}, we compare the performance of the confidence intervals from Algorithm \ref{alg:rejACC}, using the minibatch scheme to define $r_{n}(\theta)$, to the credible intervals of Algorithm \ref{alg:ISABC}, using four different choices for $\pi(\cdot)$. For IS-ABC we use an uninformative prior in setting (iv), two informative priors in settings (v) and (vi) and a misspecified prior in setting (vii), to study the case where we do not meet the conditions for a Bernstein von-Mises type of theorem. For each experiment, even though the summary statistic is not asymptotically normal, the frequentist coverage using Algorithm \ref{alg:rejACC} is closer to the $95\%$ nominal level than the coverage of the credible intervals from Algorithm \ref{alg:ISABC}, especially when the prior is misspecified as in setting (vii). Furthermore, the approximate confidence distribution intervals are more efficient than the credible intervals resulting from approximate Bayesian computing, the former having widths about half the widths of the latter except in case (v) where the prior is highly informative. In case (v), although both confidence intervals have similar width, the latter shows more  overcoverage.

\begin{table}
\caption{Coverage proportions and the median width/volume of confidence or credible intervals/regions, calculated using $300$ datasets under settings of Table~\ref{EX1_setting}. For credible intervals, both the frequentist coverage proportions and the Bayesian coverage probabilities are reported, the latter are given in the parenthesis. Each dataset contains $400$ observations, and in each algorithm run, a Monte Carlo sample of size $10^5$ is simulated. The nominal level is $95\%$.}
\begin{tabnote} \normalsize
		(A) Using an informative summary statistics for $\theta$ and $\tau$.
\end{tabnote}
		\begin{tabular}{ccccccc}
			&  & \multicolumn{2}{c}{r-ACC} &  & \multicolumn{2}{c}{IS-ABC} \\
			Setting & Acceptance  & Coverage & Width/ &  & Coverage & Width/\\
			          &proportion     &                  & Volume &  &               &Volume    \\
			\multirow{2}{*}{(i) $\theta$/ Median} & $0.005$ & $0.947$ & $0.162$  & \multirow{2}{*}{} & $0.950$ ($0.955$) & $0.169$ \\
			& $0.1$ & $0.947$ & $0.165$ &  & $0.950$ ($0.957$) & $0.17$ \\
			& $0.4$ & $0.947$ & $0.166$ &  & $0.950$ ($0.958$) & $0.17$ \\
			&  &  &  &  &  & \\
			\multirow{2}{*}{(ii) $\tau$/ MAD} & $0.005$ & $0.950$ & $0.163$ & \multirow{2}{*}{} & $0.947$ ($0.955$) & $0.169$    \\
			& $0.1$ & $0.937$ & $0.165$ &  & $0.950$ ($0.958$) & $0.170$ \\
			& $0.4$ & $0.943$ & $0.164$ &  & $0.950$ ($0.957$) & $0.171$ \\
			&  &  &  &  &  & \\
			\multirow{2}{*}{(iii) $(\theta,\tau)$ / (Median,MAD)} & $0.005$ & $0.913$ &                   $0.059$            & & $0.917$& $0.059$  \\
			& $0.1$ & $0.933$  & $0.100$   &  & $0.92$ & $0.100$ \\
			& $0.4$ & $0.94$ & $0.141$  &  & $0.927$  & $0.141$\\
		\end{tabular}
%ACC 0.059 0.09951333 0.14095330
%ABC 0.059  0.100   0.141  
		\begin{tabnote} \normalsize
			(B) Using an un-informative summary statistic for $\theta$, i.e. $S_{n}= \bar{x}$.
		\end{tabnote}
		\begin{tabular}{ccccccc}
			&  & \multicolumn{2}{c}{r-ACC} &  & \multicolumn{2}{c}{IS-ABC}\\
			Setting & Acceptance proportion & Coverage & Width &  & Coverage & Width\\
			&  &  &  &  &  & \\
			\multirow{3}{*}{(iv) $1_{\theta\in\mathbb{R}}$} & $0.005$ & $0.970$ & $2.56$ & \multirow{3}{*}{} & $0.983$ ($1$) & $4.65$\\
			& $0.1$ & $0.973$ & $2.56$ &  & $0.973$ ($1$) & $5.39$\\
			& $0.4$ & $0.963$ & $2.65$ &  & $0.967$ ($1$) & $5.58$\\
			&  &  &  &  &  & \\
			\multirow{3}{*}{(v) $t_{4}(\theta_{0},1)$} & $0.005$ & $0.970$ & $2.56$ & \multirow{3}{*}{} & $1$ ($1$) & $2.69$\\
			& $0.1$ & $0.973$ & $2.56$ &  & $1$ ($1$) & $2.65$\\
			& $0.4$ & $0.963$ & $2.65$ &  & $1$ ($1$) & $2.76$\\
			&  &  &  &  &  & \\
			\multirow{3}{*}{(vi) $t_{4}(\theta_{0},3)$} & $0.005$ & $0.970$ & $2.56$ & \multirow{3}{*}{} & $1$ ($1$) & \multirow{1}{*}{$3.93$}\\
			& $0.1$ & $0.973$ & $2.56$ &  & $1$ ($1$) & \multirow{1}{*}{$4.32$}\\
			& $0.4$ & $0.963$ & $2.65$ &  & $1$ ($1$) & $4.42$\\
			&  &  &  &  &  & \\
			\multirow{3}{*}{(vii) $t_{4}(\theta_{0}+3,3)$} & $0.005$ & $0.970$ & $2.56$ &  & $0.93$ ($1$) & $4.40$\\
			& $0.1$ & $0.973$ & $2.56$ &  & $0.89$ ($1$) & $5.33$\\
			& $0.4$ & $0.963$ & $2.65$ &  & $0.89$ ($1$) & $5.61$\\
		\end{tabular}
	\label{EX1_results}
\end{table}

\subsection{Ricker model}
A Ricker map is a non-linear dynamical system, often used in Ecology, that describes how a population changes over time. The population $N_t$ is noisily observed and is described by the following model,
\begin{align*}
& y_t \sim \mbox{Pois}(\phi N_t),\\
& N_t = r N_{t-1}e^{-N_{t-1}+e_t}, e_t\sim N(0,\sigma^2),
\end{align*}
where $t=1,\ldots,T$. Parameters $r$, $\phi$ and $\sigma$ are positive constants, interpreted as the intrinsic growth rate of the population, a scale parameter and the environmental noise. This model is statistically challenging since its likelihood function is intractable when $\sigma$ is non-zero and highly irregular in certain regions of the parameter space. \cite{wood2010statistical} suggests a summary statistic-based inference, instead of likelihood-based inference, to overcome the noise-driven nature of the model. \cite{Fearnhead2012} applies Algorithm \ref{alg:ISABC} with the regression adjustment on the above model. In this section, we apply Algorithm \ref{alg:rejACC} with the regression adjustment and compare its performance with that of regression-adjusted Algorithm \ref{alg:ISABC}.

We consider inference on the unknown parameter $\theta=(r,\phi,\sigma)$. 
A total of four different methods are compared. 
(i) Algorithm \ref{alg:rejACC} with the regression adjustment; (ii) Algorithm \ref{alg:ISABC} with the regression adjustment; 
both using Algorithm \ref{alg:MB_rn} to choose $r_{n}(\theta)$. (iii) Algorithm \ref{alg:rejACC} with the regression adjustment; (iv) Algorithm \ref{alg:ISABC} with the regression adjustment; 
both using Algorithm \ref{alg:MB_rn} with the refinement to choose $r_{n}(\theta)$. 
The main computational cost of all four algorithms is associated with the calculation of the point estimate in Algorithm \ref{alg:MB_rn}, for which we select the maximum synthetic likelihood estimator as defined in \cite{wood2010statistical}. Because each point estimate requires the simulation of a Markov chain Monte Carlo sample for the synthetic likelihood, each of the four algorithms spend over $50\%$ of CPU time on obtaining $r_{n}(\theta)$. Relative to this cost, the additional cost of the population Monte Carlo algorithm in the refined-minibatch scheme is negligible when using $10^4$ particles and $10$ iterations run in parallel. In this example, the parametric bootstrap method is not feasible due to the large number of point estimates it would need to calculate.

Following the settings used in \cite{wood2010statistical}, our dataset contains observations from $t=51$ to $100$, generated using parameter value $\theta=(e^{3.8},0.3,10)$, and using the same summary statistic therein. We assume $\theta$ follows an improper uniform prior distribution over all positive values. In Algorithm \ref{alg:MB_rn}, each minibatch has size $10$ and a total number of $40$ batches are used. They are chosen with overlaps in order to ensure a reasonable number of point estimates are available in the current small data size setting. Results are given in Table \ref{EX2_results}. Because the regression adjustment methods are better in all cases, to save time and space we only report here results for regression adjustment methods. The simulation results without the minibatch refinement, show that IS-ABC has somewhat better coverage than r-ACC since the point estimates (and thus $r_n(\cdot)$) are biased in the small data size setting. However, with the refined-minibatch scheme,  the width of the confidence intervals for r-ACC are smaller than those in IS-ABC in all cases, although both methods are over-coverage (here the target is $0.95$). This result illustrates the benefit of improving $r_{n}(\theta)$ through the population Monte Carlo procedure on problems with poor initial choice of $r_{n}(\theta)$. In the Cauchy example above, using the refined-minibatch scheme would improve upon the results however the improvement would be minimal and not as strong as in the Ricker example.

\begin{table}
	\caption{Coverage proportions and the median width of confidence/coverage intervals calculated using $150$ datasets for the four different methods of the Ricker model in Example 2 with $\delta=3/5$ for r-ACC and a flat prior for IS-ABC. Each dataset contains $50$ observations, and in each algorithm run, a Monte Carlo sample of size $10^6$ is simulated. The nominal level is $95\%$.}	
		\begin{tabnote} \normalsize
		(A) Using Algorithm \ref{alg:MB_rn} to construct $r_{n}(\theta)$
		\end{tabnote}
	\begin{tabular}{ccccccc}
			&  & \multicolumn{2}{c}{r-ACC} &  & \multicolumn{2}{c}{IS-ABC}\\
			& Acceptance proportion & Coverage & Width &  & Coverage & Width\\
			\multirow{3}{*}{$\log R$} & $0.005$ & $0.91$ & $0.59$ & \multirow{3}{*}{} & $0.91$ & $0.72$\\
			& $0.1$ & $0.91$ & $0.59$ &  & $0.99$ & $0.89$\\
			& $0.4$ & $0.9$ & $0.61$ &  & $0.99$ & $0.99$\\
			&  &  &  &  &  & \\
			\multirow{3}{*}{$\log\sigma$} & $0.005$ & $0.96$ & $2.46$ & \multirow{3}{*}{} & $0.95$ & $2.59$\\
			& $0.1$ & $0.95$ & $2.78$ &  & $0.96$ & $2.90$\\
			& $0.4$ & $0.94$ & $2.9$ &  & $0.97$  & $2.89$\\
			&  &  &  &  &  & \\
			\multirow{3}{*}{$\log\phi$} & $0.005$ & $0.89$ & $0.21$ & \multirow{3}{*}{} & $0.92$  & $0.24$\\
			& $0.1$ & $0.91$ & $0.21$ &  & $0.94$ & $0.30$\\
			& $0.4$ & $0.91$ & $0.23$ &  & $0.97$ & $0.33$\\
%			&  &  &  &  &  & \\
%			\multirow{3}{*}{($\log R$, $\log\sigma$)} & $0.005$ & $0.76$ &  &  & $0.95$ & \\
%			& $0.1$ & $0.77$ &  &  & $0.98$ & \\
%			& $0.4$ & $0.75$ &  &  & $0.99$ & \\
%			&  &  &  &  &  & \\
%			\multirow{3}{*}{($\log R$, $\log\phi$)} & $0.005$ & $0.92$ &  &  & $0.90$ & \\
%			& $0.1$ & $0.93$ &  &  & $0.91$  & \\
%			& $0.4$ & $0.92$ &  &  & $0.95$  & \\
		\end{tabular}
	\begin{tabnote} \normalsize
		(B) Using the refined version of Algorithm \ref{alg:MB_rn} to construct $r_{n}(\theta)$.
	\end{tabnote}
		\begin{tabular}{ccccccc}
			&  & \multicolumn{2}{c}{r-ACC} &  & \multicolumn{2}{c}{IS-ABC}\\
			& Acceptance proportion & Coverage & Width &  & Coverage & Width\\
			\multirow{3}{*}{$\log R$} & $0.005$ & $0.96$ & $0.85$ & \multirow{3}{*}{} & $0.97$ & $0.95$\\
			& $0.1$ & $0.99$ & $0.97$ &  & $0.99$ & $1.24$\\
			& $0.4$ & $1.00$ & $1.17$ &  & $0.99$ & $1.96$\\
			&  &  &  &  &  & \\
			\multirow{3}{*}{$\log\sigma$} & $0.005$ & $0.96$ & $1.3$ & \multirow{3}{*}{} & $0.97$ & $1.63$\\
			& $0.1$ & $0.97$ & $1.37$ &  & $0.99$ & $1.92$\\
			& $0.4$ & $1.00$ & $1.51$ &  & $0.99$  & $2.29$\\
			&  &  &  &  &  & \\
			\multirow{3}{*}{$\log\phi$} & $0.005$ & $0.96$ & $0.28$ & \multirow{3}{*}{} & $0.97$  & $0.31$\\
			& $0.1$ & $0.99$ & $0.35$ &  & $0.99$ & $0.43$\\
			& $0.4$ & $0.98$ & $0.55$ &  & $1.00$ & $0.86$\\
%			&  &  &  &  &  & \\
%			\multirow{3}{*}{($\log R$, $\log\sigma$)} & $0.005$ & $0.96$ &  &  & $0.98$ & \\
%			& $0.1$ & $0.97$ &  &  & $0.95$ & \\
%			& $0.4$ & $0.99$ &  &  & $0.99$ & \\
%			&  &  &  &  &  & \\
%			\multirow{3}{*}{($\log R$, $\log\phi$)} & $0.005$ & $0.99$ &  &  & $0.99$ & \\
%			& $0.1$ & $1.00$ &  &  & $1.00$  & \\
%			& $0.4$ & $0.99$ &  &  & $1.00$  & \\
		\end{tabular}
	\label{EX2_results}
\end{table}

\section{Discussion}   
\label{sec:discuss}
In this paper, we re-frame the well-studied popular approximate Bayesian computing method within a frequentist context and justify its performance by standards set on the frequency coverage rate. In doing so, we develop a new computational technique called {\it approximate confidence distribution computing}, a likelihood-free method that does not depend on any Bayesian assumptions such as prior information. Rather than compare the output to a target posterior distribution, the new method quantifies the uncertainty in estimation by drawing upon a direct connection to a confidence distribution. This connection guarantees that confidence intervals/regions based on approximate confidence distribution computing methods attain the frequentist coverage property even in cases where one has a finite sample size and the cases when the summary statistic used in the computing is not sufficient. 
%capture the truth about the parameters of interest at the nominal level asymptotically. 
Thus we provide theoretical support for inference from approximate confidence distribution methods which include, but are not limited to, the special case where we do have prior information (i.e. approximate Bayesian computing). 
Furthermore, in the case where the selected summary statistic is sufficient, inference based on the results of Algorithm \ref{alg:rejACC} is equivalent to maximum likelihood inference. 
%\Minge{[Comment: mentioned this in the main text too]}. 
In addition to providing sound theoretical results for inference, the framework of approximate confidence distribution computing sets the user up for better computational performance by allowing the data to drive the algorithm through the choice of $r_{n}(\theta)$. The potential computational advantage of our method has been illustrated through simulation examples.  

%When the chosen

Different choices of summary statistics often lead to different approximate Bayesian computed posteriors  $\pi_{\veps}(\theta\mid s_{\rm obs})$ in Algorithms~\ref{alg:rejABC} and \ref{alg:ISABC} and different approximate confidence distribution  $q_{\veps}(\theta\mid s_{\rm obs})$ in
Algorithm~\ref{alg:rejACC}. 
We find the philosophical interpretation of the results admitted through approximate confidence distribution computing to be more natural than the Bayesian interpretation of approximate Bayesian computed posteriors. Within a frequentist setting, it makes sense to view the many different potential confidence distributions produced by our method resulting from different choices of summary statistics as various choices of (distribution) estimators. However, within the Bayesian framework, there is no clear way to choose from among the different approximate posteriors due to various choices of summary statistics. In particular, there is an ambiguity in defining the probability measure on the joint space $(\P, {\cal X})$ when choosing among different approximate Bayesian computed posteriors. Rather than engaging in a pursuit to define a moving target such as this, our method maintains a clear frequentist interpretation thereby offering a consistently cohesive interpretation of likelihood-free methods.

In Section~\ref{sec:rn}, one may wonder if an estimate, $\widehat{\theta},$ can be computed, then why not apply the parametric bootstrap method to construct confidence regions for $\theta$ as opposed to using Algorithm \ref{alg:rejACC}? Although no likelihood evaluation is needed, this bootstrap method has two drawbacks. First, the parametric bootstrap method is heavily affected by the quality of $\widehat{\theta}$. For example, a bootstrapped confidence interval is based on quantiles of $\widehat{\theta}$ from simulated datasets. A poor estimator $\widehat{\theta}$ typically leads to poor performing confidence sets. In contrast, in Section~\ref{sec:rn}, 
% Algorithm \ref{alg:rejACC}, 
$\widehat{\theta}$ is only used to construct the initial function estimate which is then updated by the data. Second, when it is more expensive to obtain $\widehat{\theta}$ than the summary statistic, the parametric bootstrap method is computationally more costly than Algorithm~\ref{alg:rejACC}, since $\widehat{\theta}$ needs to be calculated for each pseudo dataset. Example 4.2 in Section~\ref{sec:ex} provided an example of this type of scenario.

The function $r_{n}(\theta)$ serves as the role of an initial `distributional estimate'.
% and Algorithm \ref{alg:rejACC} can be viewed as an updating mechanism to obtain a better 'distributional estimate' that incorporates data information 
Even in the instance where $r_{n}(\theta)$ does not yield reasonable acceptance probabilities for Algorithm \ref{alg:rejACC},
% \Minge{[Comment: revise this paragraph in a more positive narrative]}
many of the established techniques used in approximate Bayesian computing 
%to improve computational efficiency
can be adapted naturally to Algorithm \ref{alg:rejACC} to improve computational performance.
% by replacing the prior distribution therein with $r_{n}(\theta)$, 
For example, the likelihood-free Markov chain Monte Carlo \cite[]{marjoram2003markov} and the dimension-reduction methods on the summary statistics \cite[]{Fearnhead2012}, among others, can improve Algorithm \ref{alg:rejACC} without sacrificing the inferential guarantees explored in this paper. Furthermore, these variants of Algorithm \ref{alg:rejACC} will be more efficient than the corresponding variants of Algorithm \ref{alg:rejABC}, since $r_{n}(\theta)$ is less dispersed than the prior distribution. 

%---------------------------------------------------------------------------------------------------------------------------------		
	\section*{Acknowledgment}
	% Acknowledgments should appear after the body of the paper but before any appendices and be as brief as possible subject to politeness. Information, such as contract numbers, of no interest to readers, must
%	be excluded.
	% The authors gratefully acknowledge the support from the US National Science Foundation 
	% through grants \# DMS 151348, 1737857, 1812048. 
	The research is supported in part by research grants from the US National Science Foundation (DMS151348, 1737857, 1812048). 
	The first author also acknowledges the generous graduate support from Rutgers University.
	
	\section*{Supplementary material}

	Further instructions will be given when a paper is accepted.

	\appendix
	\section*{Appendix 1}
		\subsection*{Example of a confidence distribution}
		Consider the following example taken from \cite{Singh2007}. Suppose $X_1,\dots,X_n$ is a sample from $N(\mu, \sigma^2)$ where both $\mu$ and $\sigma^2$ are unknown. A confidence distribution for parameter $\mu$ is the function $H_n(y) = F_{t_{(n-1)}}\left\{(y-\bar{X})/(s_{n}/\sqrt{n})  \right\}$ where $F_{t_{(n-1)}}(\cdot)$ is the cumulative distribution function of a Student's t-random variable with $n-1$ degrees of freedom and $\bar{X}$ and $s_n^2$ are the sample mean and variance, respectively. Here  $H_n(y)$ is a cumulative distribution function in the parameter space of $\mu$ from which we can construct confidence intervals of $\mu$ at all levels.  For example, for any $\alpha \in (0,1)$, one sided confidence intervals for $\mu$ are $(\infty, H_n^{-1}(\alpha)]$ and $[H_n^{-1}(\alpha), \infty).$
		Similarly, a confidence distribution for parameter $\sigma^2$ is the function $H_n(\sigma^2) = 1 - F_{\chi^2_{n-1}} \left[\{(n-1)s_n^2\}/(\sigma^2) \right],$ where $ F_{\chi^2_{n-1}}(\cdot)$ is the distribution function of a Chi-squared random variable with $n-1$ degrees of freedom. 
		Again, $H_n(\sigma^2)$ is a cumulative distribution function in the parameter space of $\sigma^2$ from which we can construct confidence intervals of $\sigma$ at all levels.  
	
		\subsection*{Lemma~\ref{ABClemma}} 
		\begin{proof} The density of $\pi_{\veps}$ can be expressed by
		\begin{align*}
		 \pi_{\veps}(\theta|s_{\rm obs}) &\propto  \int_{\mathbb{R}^{d}}\pi(\theta)f_n(s\mid\theta) K_{\veps}(s-s_{\rm obs})ds \\
		 &= \pi(\theta)\int \big\{ f_n(s_{\rm obs}\mid\theta) + Df_n(\bar{s}\mid\theta)^{T}(\bar{s}-s) \\
		 &\hspace{1.5cm}+ (1/2)(\bar{s} - s)^{T}Hf_n(\bar{s}\mid\theta)(\bar{s}-s)  \big\} K_{\veps}(s-s_{\rm obs})ds\\
		 &\propto \pi(\theta)f_n(s_{\rm obs}|\theta)+ O(\veps^2),
		\end{align*}
	
		where $ D f_n(\cdot \mid\theta)$ and $H f_n(\cdot \mid\theta)$ are the vector of first derivatives and matrix of second derivatives of $ f_n(\cdot \mid\theta)$, respectively, and $\bar s$ is a value/vector between $s_{\rm obs}$ and $s_{\rm obs} +  u \veps$.  
		The equality above holds due to a Taylor expansion of $f_n(\cdot \mid\theta)$ with respect to $s_{\rm obs}$ and the final proportion holds using the substitution  $ u = (s - s_{\rm obs})$ and that $\int_{\mathbb{R}^{d}} K_{\veps}( u)\,d u = 1$ and $\int_{\mathbb{R}^{d}} u K_{\veps}( u)\,d u = 0$.  
		\end{proof}
	
			\subsection*{Remark~\ref{remk1} in Section 2} 
			\begin{proof} By its definition, $H_n(\cdot) = H(\cdot, s_{\rm obs})$ is a sample-dependent cumulative distribution function on the parameter space. 
			We also have $H_n(\theta_0)  = H(\theta_0, s_{\rm obs}) 
			= \text{pr}^*(2\htheta - \theta \leq \theta_0 \mid S_n= s_{\rm obs} )
			= \text{pr}^*(\theta - \htheta \geq \htheta- \theta_0 \mid   S_n=s_{\rm obs} )
			= 1 - G(\htheta - \theta_0)$.  Since $G(t) = \text{pr}(\htheta - \theta \leq t \mid \theta = \theta_0)$, we have $G(\htheta - \theta_0) \sim Unif(0,1)$ under the probability measure of the random sample population. Thus,  as a function of the random $ S_n$, $H_n(\theta_0) = H_n(\theta_0, S_n) \sim Unif(0,1)$. By the univariate confidence distribution definition, $H_n(\cdot)$ is a confidence distribution function.  
			
			Furthermore, $H_n(\cdot)$ can provide us confidence intervals of any level. In particular, for any $\alpha \in (0,1)$, $
			\text{pr}\{\theta \leq H_n^{-1}(1-\alpha) \mid \theta = \theta_0 \}= \text{pr}\{H_n(\theta) \leq 1- \alpha \mid \theta = \theta_0\} =1- \alpha$. Thus, $(- \infty,  H_n^{-1}(1-\alpha)]$ is a $(1-\alpha)$-level confidence interval. Note that, $H_n(2\htheta- \theta_{\alpha}) = \text{pr}^*(2\theta_{ACC} - \theta \leq 2\theta- \theta_{\alpha}\mid S_n =  s_{\rm obs} )  = 1 - \text{pr}^*(\theta < \theta_{\alpha}\mid   S_n = s_{\rm obs} ) = 1 - \alpha$. So, $H_n^{-1}(1-\alpha) = 2\htheta- \theta_{\alpha}$. Therefore, $(- \infty,  2\htheta- \theta_{\alpha}]$ is also a $(1-\alpha)$-level confidence interval for $\theta$.
			\end{proof}
		
		\subsection*{Lemma~\ref{main1}} 
		\begin{proof}
		First note that 
		\bes 
		&&\mid \text{pr}\{\theta \in \Gamma_{1 - \alpha}(S_{n}) \mid \theta = \theta_0 \}  - (1 - \alpha)   \mid  =   \mid \text{pr}\{W( \theta, S_n) \in A_{1 - \alpha} \mid  \theta = \theta_0 \}  -  (1 - \alpha)  \mid  \\
		&\leq&   \mid  \text{pr}^*\{V(\theta, S_n) \in A_{1 - \alpha} \mid  S_n=s_{\rm obs}\} - (1 - \alpha)\mid  \\
		&&\hspace{1cm} +  \mid  \text{pr}\{W( \theta, S_n) \in A_{1 - \alpha} \mid  \theta = \theta_0 \} 
		 - \text{pr}^*\{V(\theta, S_n) \in A_{1 - \alpha} \mid  S_n = s_{\rm obs} \}
		\mid 
		\ees 
		and by the definition of $A_{1 - \alpha}$ in (4), $\mid \text{pr}^*\{V(\theta, S_n) \in A_{1 - \alpha} \mid  S_n = s_{\rm obs}\} - (1 - \alpha)\mid   = o(\delta')$, almost surely for a pre-selected precision number, $\delta'>0$. Therefore, by Condition \ref{cond:ACC_interval}, we have
		$\mid  \text{pr}\{\theta \in \Gamma_{1 - \alpha}(S_{n}) \mid \theta = \theta_0 \} - (1 - \alpha)\mid = \delta$
		% \leq o_p(\delta) + o_p(1) = o_p(1)$ and hence $ \mid  \text{pr}\{\theta_0 \in \Gamma_{1 - \alpha}(S_n) \} - (1 - \alpha)\mid  = o_p(\delta)$ for 
		where $\delta = \max\{\delta_{\veps},\delta'\}$. Furthermore, if Condition \ref{cond:ACC_interval} holds almost surely, then $\mid \text{pr}\{\theta \in \Gamma_{1 - \alpha}(S_{n}) \mid \theta = \theta_0 \} - (1 - \alpha)\mid  = o(\delta)$, almost surely. 
	\end{proof}

	\subsection*{Theorem 1}
	\begin{proof}
		% Without loss of generality, we assume $t$ and $S_n$ have the same dimension.
		% Denote by $g(t)$, the density of the pivot statistic $T = T(S_n, \theta)$. 
		% Since the distribution for the pivot is free of parameter, the conditional distribution of  $T =T(S_n, \btheta)$, for any given $\btheta$, is also $g(t)$. 
		Setting $W(\btheta, S_{n}) = T( \btheta, S_n)$ and by (\ref{eq:apivot}), we immediately have
		\begin{equation}\label{eq:AW}
		\text{pr}\{W(\btheta, S_{n}) \in A \mid \theta = \btheta_0\} =  \int_{t \in A}
		g(t) d t \,  \{1 + o(\delta^{''})\}, 
		\end{equation}
		for any Borel set $A \subset \mathbb{R}^{d}$.

		Let $f(s|\btheta)$ be the conditional density of $S_n$, given $\btheta$.
		Note that $t$ and $S_n$ have the same dimension.
		For a given $\btheta$ and with the variable transformation $T = T( \btheta, S_n)$, the density functions $g(t)$ and $f(s_{t, \theta}|\btheta)$ are connected by a Jacobi matrix: 
		$f(s_{t, \theta}|\btheta) |T^{(1)}( \btheta, s_{t, \theta})|^{-1} = g(t) \{1 + o(\delta^{''})\}$, where $T^{(1)}( \btheta, s) = \frac{\partial}{\partial s} T(\btheta, s)$ and $s_{t, \theta}$ is the solution of $t = T( \btheta, s)$. 
		
		In Algorithm \ref{alg:rejACC}, we simulate $\btheta' \sim r_{n}(\btheta)$ and $s' = S_n(x')$ with $x' | \btheta = \btheta' \sim M_{\theta'}$. 
		% Denote by $t' = T(s', \theta')$.  
		Furthermore, we only keep those pairs $(\btheta', s')$ with the kernel probability $K_{\veps}(s'-s_{\rm obs})$. Thus, the joint density function of a copy of $(\btheta', s')$ that are simulated and kept by Algorithm \ref{alg:rejACC}, conditional on observing $S_n  = s_{\rm obs}$, is
		\begin{align*}
		(\btheta', s') | S_n  = s_{\rm obs} \propto  r_{n}(\theta') f_n(s'\mid\theta') K_{\veps}(s'-s_{\rm obs}).
		% \\ & =   r_{n}(\theta') g(t')  |T^{(1)}(s', \btheta')|K\{\varepsilon^{-1}(s'-s_{\rm obs})\}' 
		\end{align*}
		Now, let $T' = T(\theta', s')$.  Perform a variable transformation from $(\btheta', s')$ to $(\btheta', T')$ with the Jacobi term $|T^{(1)}( \btheta', s_{T', \theta'})|^{-1}$, where $s_{T',\theta'}$ is a solution to $T' = T(\theta',s)$. Then, the joint conditional density of $(\btheta', T')$, conditional on $S_n  = s_{\rm obs}$,  is
		\begin{align*}
		(\btheta', T') | S_n  = s_{\rm obs} & \propto  r_{n}(\theta') f_n(s_{T',\theta'}\mid\theta') |T^{(1)}( \btheta', s_{T', \theta'})|^{-1} K_{\veps}(s_{T', \theta'}-s_{\rm obs}).
		\\ & =   r_{n}(\theta') g(T')  K_{\veps}(s_{T', \theta'}-s_{\rm obs})\{1 + o(\delta^{''})\}.
		\end{align*}
		Therefore, $T' = T(\theta', s')$, the approximate pivot statistic generated from Algorithm \ref{alg:rejACC}, with distribution conditional on $S_n  = s_{\rm obs}$: 
		\begin{align*}
		T' | S_n  = s_{\rm obs} & \propto g(t')  \{1 + o(\delta^{''})\} \int  r_{n}(\theta')  K_{\veps}(s_{t', \theta'}-s_{\rm obs}) d \btheta'
		\end{align*}
		If requirement (\ref{eq:req}) is satisfied,  
		%and we set $V(s',  \btheta') = T' = T(s', \theta')$,  
		then we have 
		$$
		T' | S_n  = s_{\rm obs} \sim g(T') \{1 + o(\delta^{''})\} \{1 + o(\delta^{'}_\veps)\}.
		$$ 
		Set $V( \btheta',s') = T' = T( \theta', s')$ and denote by $\btheta_{\rm ACC}$ the $\btheta'$ accepted by the ACC algorithm. We have
		$$
		\text{pr}^*\{V(\btheta_{\rm ACC}, S_n) \in A \mid  S_n= s_{\rm obs} \} 
		= \int_{t \in A} g(t)  dt \{1 + o(\delta^{''})\} \{1 + o(\delta^{'}_\veps)\}
		$$
		Thus, together with (\ref{eq:AW}), Condition \ref{cond:ACC_interval} is satisfied for $\delta_\veps = \max\{\delta^{''}, \delta^{'}_{\veps}\}$.
		Furthermore, by Lemma 2, the rest of the statements in the theorem also hold.  
	\end{proof}

	\subsection*{Corollary 1}
	\begin{proof}
		%We just verify requirement (\ref{eq:req}), that $ \int  r_{n}(\sigma) K\left\{\veps^{-1} \left(s_{T, \sigma}-s_{\rm obs}  \right)  \right\}d \sigma$ is free of  $T$. 
		Here we prove requirement (\ref{eq:req}) for Part 2, data from a scale family. The proofs for Part 1 (location family) and Part 3 (location and scale family) are similar and thus omitted.
		
		In particular, in a scale family suppose $S_n$ has the density $(1/\sigma) g_2(S_n / \sigma)$. Then $T = T(\sigma, S_n) = S_n/ \sigma \sim g_2(t)$ is a pivot. So, for any given $(t,\sigma)$ pair we have $s_{t,\sigma} = t \sigma$. Thus, with variable transformation $u =  t\sigma - s_{\rm obs}$ 
		%and assuming that $K(\cdot)$ is symmetric,  
		we have
		\bes 
		\int  r_{n}(\sigma) K_{ \veps}\left(s_{T, \sigma}-s_{\rm obs}\right)  d \sigma
		&=& \int  \frac{1}{\sigma} K_{ \veps}\left(s_{T, \sigma}-s_{\rm obs}\right) d \sigma \\
		&=& \int \frac{1}{u + s_{\rm obs}} K_{\veps} \left(u + s_{\rm obs} - s_{\rm obs} \right) du 
		%= s_{\rm obs}^{-1} + o(\veps), 
		\ees 
		which is free of $t$. Therefore, the requirement (\ref{eq:req}) is satisfied in this case.
		Furthermore, the function $H_2(\hat{\sigma}^2_{S}, x) = 1 - \int_{0}^{\hat{\sigma}^{2}_{S}/x}g_2(w)dw$ is a confidence distribution for $\sigma^2$ % induced by $\hat{\sigma}^{2}_{S}/\sigma^2$,  
		since (1) given $S$, $H_2(\hat{\sigma}^2_{S}, x) $ is a distribution function on the parameter space $(0, \infty)$ and (2) given $x = \sigma^2_0$, % the true parameter value, 
		$H_2(\hat{\sigma}^2_{S}, x) \sim U(0,1)$.
	\end{proof}

\subsection*{Additional Conditions and notations for the remaining proofs}

\noindent Let $N(\bx;\bmu,\Sigma)$ be the normal density at $\bx$
with mean $\bmu$ and variance $\Sigma$, and $\ftil_{n}(\bs\mid\btheta)=N\{\bs;\bs(\btheta),A(\btheta)/a_{n}^{2}\}$, the asymptotic distribution of the summary statistic.
We define $a_{n,\veps}=a_{n}$ if $\lim_{n\rightarrow\infty}a_{n}\veps_{n}<\infty$
and $a_{n,\veps}=\veps_{n}^{-1}$ otherwise, and $\ensuremath{c_{\veps}=\lim_{n\rightarrow\infty}a_{n}\veps_{n}}$,
both of which summarize how $\veps_{n}$ decreases relative to the
converging rate, $a_{n}$, of $S_{n}$ in Condition \ref{sum_conv} below. Define
the standardized random variables $W_{n}(S_{n})=a_{n}A(\theta)^{-1/2}\{S_{n}-\eta(\theta)\}$
and $W_{\rm obs}=a_{n}A(\theta)^{-1/2}\{s_{\rm obs}-\eta(\theta)\}$
according to Condition \ref{sum_conv} below. Let $f_{W_{n}}(w\mid\theta)$ and
$\ftil_{W_{n}}(w\mid\theta)$ be the density for $W_{n}(S_{n})$ when
$S_{n}\sim f_{n}(\cdot\mid\theta)$ and $\ftil_n(\cdot\mid\theta)$ respectively.
Let $B_{\delta}=\{\theta\mid\|\theta-\theta_{0}\|\leq\delta\}$ for
$\delta>0$. Define the initial density truncated in $B_{\delta}$,
i.e. $r_{n}(\theta)\mathbb{I}_{\theta\in B_{\delta}}/\int_{B_{\delta}}r_{n}(\theta)\,d\theta$,
by $r_{\delta}(\theta)$. Let $t(\theta)=a_{n,\veps}(\theta-\theta_{0})$
and $v(s)=\veps_{n}^{-1}(s - s_{\rm obs})$. For any $A\in\mathscr{B}^{p}$
where $\mathscr{B}^{p}$ is the Borel sigma-field on $\mathbb{R}^{p}$,
let $t(A)$ be the set $\{\phi:\phi=t(\theta)\text{ for some }\theta\in A\}$.
For a non-negative function $h(x)$, integrable in $\mathbb{R}^{l}$,
denote the normalized function $h(x)/\int_{\mathbb{R}^{l}}h(x)\,dx$
by $h(x)^{({\rm norm})}$. For a function $h(x)$, denote its gradient
by $D_{x}h(x)$, and for simplicity, omit $\theta$ from $D_{\theta}$. For a sequence $x_n$, we use the notation $x_n = \Theta(a_n)$ to mean that there exist some constants $m$ and $M$ such that $0<m<\mid x_n/a_n \mid<M<\infty$.

\begin{condition} \label{sum_conv}
There exists a sequence $a_{n}$, satisfying $a_{n}\rightarrow\infty$
as $n\rightarrow\infty$, a $d$-dimensional vector $\eta(\btheta)$
and a $d\times d$ matrix $A(\btheta)$, such that for $S_{n}\sim f_{n}(\cdot\mid\theta)$
and all $\btheta\in\mathcal{P}_{0}$, 
\[
a_{n}\{\bS_{n}-\eta(\btheta)\}\rightarrow N\{0,A(\btheta)\},\mbox{ as \ensuremath{n\rightarrow\infty}},
\]
in distribution. We also assume that $\bs_{\rm obs}\rightarrow\eta(\btheta_{0})$
in probability. Furthermore, it holds that (i) $\eta(\btheta)$
and $A(\btheta)\in C^{1}(\mathcal{P}_{0})$, and $A(\btheta)$ is
positive definite for any $\btheta$; (ii) for any $\delta>0$ there
exists a $\delta'>0$ such that $\|\eta(\btheta)-\eta(\btheta_{0})\|>\delta'$
for all $\btheta$ satisfying $\|\btheta-\btheta_{0}\|>\delta$; and
(iii) $I(\theta)\triangleq\left\{ \frac{\partial}{\partial\theta}\eta(\theta)\right\} ^{T}A^{-1}(\theta)\left\{ \frac{\partial}{\partial\theta}\eta(\theta)\right\} $
has full rank at $\btheta=\btheta_{0}$.
\end{condition}

\begin{condition} \label{kernel_prop}
The kernel satisfies (i) $\int vK_{\veps}(v)dv=0$; (ii)$\prod_{k=1}^{l}v_{i_{k}}K_{\veps}(v)dv<\infty$
for any coordinates $(v_{i_{1}},\dots,v_{i_{l}})$ of $v$ and $l\leq p+6$;
(iii)$K_{\veps}(v)\propto K_{\veps}(\|v\|_{\Lambda}^{2})$ where $\|v\|_{\Lambda}^{2}=v^{T}\Lambda v$
and $\Lambda$ is a positive-definite matrix, and $K(v)$ is a decreasing
function of $\|v\|_{\Lambda}$; (iv) $K_{\veps}(v)=O(\exp\{-c_{1}\|v\|^{\alpha_{1}}\})$
for some $\alpha_{1}>0$ and $c_{1}>0$ as $\|v\|\rightarrow\infty$. 
\end{condition}

\begin{condition} \label{sum_approx}
There exists $\alpha_{n}$ satisfying $\alpha_{n}/a_{n}^{2/5}\rightarrow\infty$
and a density $r_{max}(w)$ satisfying Condition \ref{kernel_prop}(ii)--(iii) where $K_{\veps}(v)$
is replaced with $r_{max}(w)$, such that $\sup_{\theta\in B_{\delta}}\alpha_{n}\mid f_{W_{n}}(w\mid\theta)-\ftil_{W_{n}}(w\mid\theta)\mid\leq c_{3}r_{max}(w)$
for some positive constant $c_{3}$. 
\end{condition}

\begin{condition} \label{sum_approx_tail}
The following statements hold: (i) $r_{max}(w)$ satisfies
Condition \ref{kernel_prop}(iv); and (ii) $\sup_{\theta\in B_{\delta}^{C}}\ftil_{W_{n}}(w\mid\theta)=O(e^{-c_{2}\|w\|^{\alpha_{2}}})$
as $\|w\|\rightarrow\infty$ for some positive constants $c_{2}$
and $\alpha_{2}$, and $A(\theta)$ is bounded in ${\cal P}$. 
\end{condition}

\begin{condition} \label{cond:likelihood_moments}
The first two moments, $\int_{\mathbb{R}^{d}}s\ftil_{n}(s\mid\theta)ds$
and $\int_{\mathbb{R}^d}s^{T}s\ftil_{n}(s\mid\theta)ds$, exist. 
\end{condition}

\subsection*{Proof of Theorem 2}

\noindent Let $\tilde{Q}(\theta\in A\mid s)=\int_{A}r_{\delta}(\theta)\ftil_{n}(s\mid\theta)\,d\theta/\int_{\mathbb{R}^{p}}r_{\delta}(\theta)\ftil_{n}(s\mid\theta)\,d\theta$. 

\begin{lemma}\label{Alemma1} Assume Condition \ref{par_true}--\ref{sum_approx}. If $\veps_{n}=O(a_{n}^{-1})$, for any fixed $\nu\in\mathbb{R}^{d}$
and small enough $\delta$, 
\[
\sup_{A\in\mathfrak{B}^{p}}\left|\tilde{Q}\{a_{n}(\theta-\theta_{0})\in A\mid s_{\rm obs}+\veps_{n}\nu\}-\int_{A}N[t;\beta_{0}\{A(\theta_{0})^{1/2}W_{\rm obs}+c_{\veps}\nu\},I(\theta_{0})^{-1}]dt\right|\rightarrow0,
\]
in probability as $n\rightarrow\infty$, where $\beta_{0}=I(\theta_{0})^{-1}D\eta(\theta_{0})^{T}A(\theta_{0}^{-1})$.
\end{lemma}
\begin{proof}
With Lemma 1 from \cite{Li2017}, it is sufficient to show that 
\[
\sup_{A\in\mathfrak{B}^{p}}\mid\tilde{Q}\{t(\theta)\in A\mid s_{\rm obs}+\veps_{n}\nu\}-\tPi\{t(\theta)\in A\mid s_{\rm obs}+\veps_{n}\nu\}\mid=o_{P}(1),
\]
where $\tPi$ denotes $\tilde{Q}$ using $r_{n}(\theta)$ rather than a prior
$\pi(\theta)$ with a density satisfying Condition \ref{par_true}. With the transformation $t=t(\theta)$
and $v=v(s)$, the left hand side of the above equation can be written
as 
\begin{eqnarray}
&&\sup_{A\in\mathfrak{B}^{p}}\mid\frac{\int_{A}r_{\delta}(\theta+a_{n}^{-1}t)\ftil_{n}(s_{{\rm obs}}+\veps_{n}\nu\mid\theta+a_{n}^{-1}t)dt}{\int_{\mathbb{R}^{p}}r_{\delta}(\theta+a_{n}^{-1}t)\ftil_{n}(s_{{\rm obs}}+\veps_{n}\nu\mid\theta+a_{n}^{-1}t)dt}-\hfill \\ \notag
&&\hspace{3cm}\frac{\int_{A}\pi(\theta+a_{n}^{-1}t)\ftil_{n}(s_{{\rm obs}}+\veps_{n}\nu\mid\theta+a_{n}^{-1}t)dt}{\int_{\mathbb{R}^{p}}\pi(\theta+a_{n}^{-1}t)\ftil_{n}(s_{{\rm obs}}+\veps_{n}\nu\mid\theta+a_{n}^{-1}t)dt}\mid.\label{eq1}
\end{eqnarray}
 For a function $\tau:\mathbb{R}^{p}\rightarrow\mathbb{R},$ define
the following auxiliary functions,
\begin{eqnarray*}
\phi_{1}\{\tau(\theta);n\} & = & \frac{\int_{t(B_{\delta})}|\tau(\theta+a_{n}^{-1}t)-\tau(\theta)|\ftil_{n}(s_{{\rm obs}}+\veps_{n}\nu\mid\theta+a_{n}^{-1}t)\,dt}{\int_{t(B_{\delta})}\tau(\theta+a_{n}^{-1}t)\ftil_{n}(s_{{\rm obs}}+\veps_{n}\nu\mid\theta+a_{n}^{-1}t)\,dt},\\
\phi_{2}\{\tau(\theta);n\} & = & \frac{\tau(\theta)\int_{t(B_{\delta})}\ftil_{n}(s_{{\rm obs}}+\veps_{n}\nu\mid\theta+a_{n}^{-1}t)dt}{\int_{t(B_{\delta})}\tau(\theta+a_{n}^{-1}t)\ftil_{n}(s_{{\rm obs}}+\veps_{n}\nu\mid\theta_{0}+a_{n}^{-1}t)dt}.
\end{eqnarray*}
Then by adding and subtracting $\phi_{2}\{\tau_{n}^{-p}r_{\delta}(\theta);n\}\phi_{2}\{\pi(\theta);n\}$
in the absolute sign of \eqref{eq1}, \eqref{eq1} can be bounded
by 
\begin{eqnarray*}
\phi_{1}\{\tau_{n}^{-p}r_{\delta}(\theta);n\}+\phi_{1}\{\pi(\theta);n\}\phi_{2}\{\tau_{n}^{-p}r_{\delta}(\theta);n\}+\phi_{1}\{\tau_{n}^{-p}r_{\delta}(\theta);n\}\phi_{2}\{\pi(\theta);n\}+\phi_{1}\{\pi(\theta);n\}.
\end{eqnarray*}
Consider a class of function $\tau(\theta)$ satisfying the following
conditions: 

There exists a series $\{k_{n}\}$, such that $\sup_{\theta\in\mathcal{P}_{0}}\|k_{n}^{-1}D\tau(\theta)\|<\infty$
and $k_{n}=o(a_{n});$ 

$\tau(\theta_{0})>0$ and $\tau(\theta)\in C^{1}(B_{\delta}).$ 

By Conditions \ref{par_true}--\ref{initial_gradient}, $\tau_{n}^{-p}r_{\delta}(\theta)$
and $\pi(\theta)$ belong to the above class. Then if $\phi_{1}\{\tau(\theta);n\}$
is $o_{p}(1)$ and $\phi_{2}\{\tau(\theta);n\}$ is $O_{p}(1)$, \eqref{eq1}
is $o_{p}(1)$ and the lemma holds. 

First, from $(ii)$, there exists an open set $\omega\subset B_{\delta}$
such that $\inf_{\theta\in\omega}\tau(\theta)>c_{1}$, for a constant
$c_{1}>0$. Then for $\phi_{2}\{\tau(\theta);n\}$, it is bounded
by 
\[
\frac{\tau(\theta)}{c_{1}\int_{t(\omega)}\ftil_{n}(s_{{\rm obs}}+\veps_{n}\nu\mid\theta_{0}+a_{n}^{-1}t)^{(norm)}dt}.
\]
From equation (7) in the supplementary material of \cite{Li2016},
$\ftil_{n}(s_{{\rm obs}}+\veps_{n}\nu\mid\theta+a_{n}^{-1}t)$
can be written in the following form, 
\begin{eqnarray}
a_{n}^{d}\ftil_{n}(s_{{\rm obs}}+\veps_{n}\nu\mid\theta+a_{n}^{-1}t)=\frac{1}{\|B_{n}(t)\|^{1/2}}N[C_{n}(t)\{A_{n}(t)t-b_{n}\nu-c_{2}\};\theta,I_{d}],\label{eq2}
\end{eqnarray}
where $A_{n}(t)$ is a series of $d\times p$ matrix functions, $\{B_{n}(t)\}$
and $\{C_{n}(t)\}$ are a series of $d\times d$ matrix functions,
$b_{n}$ converges to a non-negative constant and $c_{2}$ is a constant,
and the minimum of absolute eigenvalues of $A_{n}(t)$ and eigenvalues
of $B_{n}(t)$ and $C_{n}(t)$ are all bounded and away from $0$.
Then for fixed $\nu$, by continuous mapping, \eqref{eq2} is away
from zero with probability one. Therefore $\phi_{2}\{\tau(\theta);n\}=O_{P}(1)$.

Second, by Taylor expansion, $\tau(\theta+a_{n}^{-1}t)=\tau(\theta)+a_{n}^{-1}D\tau(\theta+e_{t}t)t$,
where $\|e_{t}\|\leq a_{n}^{-1}$. Then 
\begin{eqnarray}
\phi_{1}\{\tau(\theta);n\} & = & \frac{k_{n}\phi_{2}\{\tau(\theta);n\}}{a_{n}\tau(\theta)}\frac{\int_{t(B_{\delta})}|k_{n}^{-1}D\tau(\theta+e_{t}t)t|\ftil_{n}(s_{{\rm obs}}+\veps_{n}\nu\mid\theta+a_{n}^{-1}t)\,dt}{\int_{t(B_{\delta})}\ftil_{n}(s_{{\rm obs}}+\veps_{n}\nu\mid\theta+a_{n}^{-1}t)\,dt}\nonumber \\
 & \leq & \frac{k_{n}\phi_{2}\{\tau(\theta);n\}}{a_{n}\tau(\theta)}\sup_{\theta\in B_{\delta}}\|k_{n}^{-1}D\tau(\theta)\|\frac{\int_{t(B_{\delta})}\|t\|a_{n}^{d}\ftil_{n}(s_{{\rm obs}}+\veps_{n}\nu\mid\theta+a_{n}^{-1}t)dt}{\int_{t(B_{\delta})}a_{n}^{d}\ftil_{n}(s_{{\rm obs}}+\veps_{n}\nu\mid\theta+a_{n}^{-1}t)\,dt},\label{eq3}
\end{eqnarray}
where the inequality holds by the triangle inequality. By the expression
\eqref{eq2} and Lemma 7 in the supplementary material of \cite{Li2016},
the right hand side of \eqref{eq3} is $O_{P}(1)$. Then together
with $\phi_{2}\{\tau(\theta);n\}=\Theta_{P}(1)$, $\phi_{1}\{\tau(\theta);n\}=o_{P}(1).$
Therefore the Lemma holds.
\end{proof}
\noindent Define the joint density of $(\theta,s)$ in Algorithm 2
and its approximation, where the s-likelihood is replaced by its Gaussian
limit and $r_{n}(\theta)$ by its truncation, by $q_{\veps}(\theta,s)$
and $\tilde{q}_{\veps}(\theta,s)$. It is easy to see that,
\begin{align*}
q_{\veps}(\theta,s) & =\frac{r_{n}(\theta)f_{n}(s|\theta)K_{\veps_{n}}(s-s_{\rm obs})}{\int_{\mathbb{R}^{p}\times\mathbb{R}^{d}}r_{n}(\theta)f_{n}(s|\theta)K_{\veps_{n}}(s-s_{\rm obs})\,d\theta ds},\\
\tilde{q}_{\veps}(\theta,s) & =\frac{r_{\delta}(\theta)\ftil_{n}(s|\theta)K_{\veps_{n}}(s-s_{\rm obs})}{\int_{\mathbb{R}^{p}\times\mathbb{R}^{d}}r_{\delta}(\theta)\ftil_{n}(s|\theta)K_{\veps_{n}}(s-s_{\rm obs})\,d\theta ds}.
\end{align*}
Let $\tilde{Q}_{\veps}(\theta\in A\mid s_{\rm obs})$ be the approximate confidence distribution function, $\int_{A}\int_{\mathbb{R}^{d}}\tilde{q}_{\veps}(\theta,s)\,dsd\theta$.
With the transformation $t=t(\theta)$ and $v=v(s)$, let $\tilde{q}_{\veps,t\nu}(t,v)=\tau_{n}^{-p}r_{\delta}(\theta+a_{n,\veps}^{-1}t)\ftil_{n}(s_{\rm obs}+\veps_{n}\nu\mid\theta+a_{n,\veps}^{-1}t)K_{\veps}(\nu)$
be the transformed and unnormalized $\tilde{q}_{\veps}(\theta,s)$, and
$\tilde{q}_{A,tv}(h)=\int_{A}\int_{\mathbb{R}^{d}}h(t,v)\tilde{q}_{\veps,t\nu}(t,v)\,dvdt$
for any function $h(\cdot,\cdot)$ in $\mathbb{R}^{p}\times\mathbb{R}^{d}$.
Denote the factor of $\tilde{q}_{\veps,t\nu}(t,v)$, $\tau_{n}^{-p}r_{\delta}(\theta+a_{n,\veps}^{-1}t)$,
by $\gamma_{n}(t)$. Let $\gamma=\lim_{n\rightarrow\infty}\tau_{n}^{-p}r_{\delta}(\theta)$
and $\gamma(t)=\lim_{n\rightarrow\infty}\tau_{n}^{-p}r_{\delta}(\theta+\tau_{n}^{-1}t)$,
the limits of $\gamma_{n}(t)$ when $a_{n,\veps}=a_{n}$ and
$a_{n,\veps}=\tau_{n}$ respectively. By Condition \ref{initial_upper} and \ref{initial_lower},
$\gamma(t)$ exists and $\gamma$ is non-zero with positive probability.
Here several functions of $t$ and $v$ defined in \cite[proofs for Section 3.1]{Li2017}
and relate to the limit of $\tilde{q}_{\veps,t\nu}(t,v)$ are used, including
$g(v;A,B,c)$, $g_{n}(t,v)$, $G_{n}(v)$ and $g_{n}'(t,v)$. Furthermore
several functions defined by integration as following are used: for
any $A\in\mathfrak{B}^{p}$, let $g_{A,r}(h)=\int_{\mathbb{R}^{d}}\int_{t(A)}h(t,v)\gamma_{n}(t)g_{n}(t,v)\,dtdv$,
$G_{n,r}(v)=\int_{t(B_{\delta})}\gamma_{n}(t)g_{n}(t,v)\,dt$, $q_{A}(h)=\int_{A}\int_{\mathbb{R}^{d}}h(\theta,s)r_{n}(\theta)f_{n}(s\mid\theta)K_{\veps}(s-s_{\rm obs})\veps_{n}^{-d}\,dsd\theta$
and $\tilde{q}_{A}(h)=\int_{A}\int_{\mathbb{R}^{d}}h(\theta,s)r_{\delta}(\theta)\ftil_{n}(s\mid\theta)K_{\veps}(s-s_{\rm obs})\veps_{n}^{-d}\,dsd\theta$,
which generalize those defined in \cite[proofs for Section 3.1]{Li2017}
for the case $r_{n}(\theta)=\pi(\theta)$. 

\begin{lemma}\label{Alemma2} Assume Condition \ref{par_true}--\ref{kernel_prop}. If $\veps_{n}=o(a_{n}^{-1/2})$, then 
\begin{enumerate}
\item[(i)] $\int_{\mathbb{R}^{d}}\int_{t(B_{\delta})}|\tilde{q}_{\veps,t\nu}(t,\nu)-\gamma_{n}(t)g_{n}(t,\nu)|\,dtd\nu=o_{p}(1)$;
\item[(ii)] $g_{B_{\delta},r}(1)=\Theta_{P}(1);$ 
\item[(iii)] $\tilde{q}_{B_{\delta},tv}(t^{k_{1}}v^{k_{2}})/\tilde{q}_{B_{\delta},tv}(1)=g_{B_{\delta},r}(t^{k_{1}}v^{k_{2}})/g_{B_{\delta},r}(1)+O_{P}(a_{n,\veps}^{-1})+O_{P}(a_{n}^{2}\veps_{n}^{4})$
for $k_{1}$ and $k_{2}$ \textcolor{red}{??} 
\item[(iv)] $\tilde{q}_{B_{\delta}}(1)=\tau_{n}^{p}a_{n,\veps}^{d-p}\left\{ \int_{t(B_{\delta})}\int_{\mathbb{R}^{d}}\gamma_{n}(t)g_{n}(t,\nu)d\tau d\nu+O_{P}(a_{n,\veps}^{-1})+O_{P}(a_{n}^{2}\veps_{n}^{4})\right\} $. 
\end{enumerate}
\end{lemma}
\begin{proof}
These results generalize Lemma 2 in \cite{Li2017} and Lemma 5 in
\cite{Li2016}. In Lemma 2 of \cite{Li2017} where $\gamma_{n}(t)=\pi(\theta+a_{n,\veps}^{-1}t)$,
$(i)$ holds by expanding $\tilde{q}_{\veps,t\nu}(t,v)$ according to
the proof of Lemma 5 of \cite{Li2016}. Here the lines can be followed
similarly by changing the terms involving $\pi(\theta)$ in equations
(10) and (11) in the supplements of \cite{Li2016}. Equation (10)
is replaced by 
\[
\frac{\gamma_{n}(t)}{\mid A(\theta+a_{n,\veps}^{-1}t)\mid^{1/2}}=\frac{\gamma_{n}(t)}{\mid A(\theta)\mid^{1/2}}+a_{n,\veps}^{-1}\gamma_{n}(t)D\frac{1}{\mid A(\theta+e_{t})\mid^{1/2}}t,
\]
where $\|e_{\tau}\|\leq\delta$, and this leads to replacing $\pi(\theta)\int_{\tau(B_{\delta})\times\mathbb{R}^{d}}g_{n}(t,\nu)dtd\nu$
in equation (11) by $\int_{\tau(B_{\delta})\times\mathbb{R}^{d}}\gamma_{n}(t)g_{n}(t,\nu)dtd\nu$.
These changes have no effect on the arguments therein since $\sup_{t\in t(B_{\delta})}\gamma_{n}(t)=O_{P}(1)$
by Condition \ref{initial_upper}. Therefore $(i)$ holds.

For (ii), By Condition \ref{initial_lower} and Lemma 2 of \cite{Li2017}, there exists a $\delta'<\delta$
such that $\inf_{t\in t(B_{\delta'})}\gamma_{n}(t)=\Theta_{p}(1)$
and $\int_{\mathbb{R}^{d}}\int_{t(B_{\delta'})}g_{n}(t,\nu)\,dtdv=\Theta_{p}(1)$.
Then since $g_{B_{\delta},r}(1)\geq\inf_{t\in t(B_{\delta'})}\gamma_{n}(t)\int_{\mathbb{R}^{d}}\int_{t(B_{\delta'})}g_{n}(t,\nu)\,dtd\nu$,
(ii) holds.

For $(iii)$, $\tilde{q}_{B_{\delta},tv}(t)/\tilde{q}_{B_{\delta},tv}(1)$
can be expanded by following the arguments in the proof of Lemma 5
of \cite{Li2016}. For $\tilde{q}_{B_{\delta},tv}(t^{k_{1}}v^{k_{2}})/\tilde{q}_{B_{\delta},tv}(1)$,
it can be expanded similarly as in the proof of Lemma 4 of \cite{Li2017}.

For $(iv)$, $\gamma_{n}(t)$ plays the same role as $\pi(\theta)$
in the proof of Lemma 5 in \cite{Li2016}, and the arguments therein
can be followed exactly. The term $\tau_{n}^{p}$ is from the definition
of $\gamma_{n}(t)$ that $r_{n}(\theta+a_{n,\veps}^{-1}t)=\tau_{n}^{p}\gamma_{n}(t)$.
\end{proof}
Define the expectation of $\theta$ with distribution $\tilde{Q}_{\veps}(\theta\in A\mid s_{\rm obs})$
as $\ttheta_{\veps}$ , and that of $\theta_{{\rm ACC}}^{*}$
with density $\tilde{q}_{\veps}(\theta,s)$ as $\ttheta_{\veps}^{*}$.
Let $E_{G,r}(\cdot)$ be the expectation with the density $G_{n}(v)^{({\rm norm})}$,
and $E_{G,r}\{h(v)\}$ can be written as $g_{B_{\delta},r}\{h(v)\}/g_{B_{\delta},r}(1)$.
Let $\psi(\nu)=k_{n}^{-1}\beta_{0}\{A(\theta_{0})^{1/2}W_{\rm obs}+a_{n}\veps_{n}\nu\}$,
where $k_{n}=1$, if $c_{\veps}<\infty$, and $a_{n}\veps_{n}$,
if $c_{\veps}=\infty$. 

\begin{lemma}\label{Alemma3} Assume Condition \ref{par_true}--\ref{initial_gradient} and \ref{kernel_prop}. Then if $\veps_{n}=o(a_{n}^{-1/2})$, 
\begin{enumerate}
\item[(i)] $\ttheta_{\veps}=\theta_{0}+a_{n}^{-1}\beta_{0}A(\theta_{0})^{1/2}W_{\rm obs}+\veps_{n}\beta_{0}E_{G_{n},r}(\nu)+r_{1}$,
where $r_{1}=o_{P}(a_{n}^{-1})$; 
\item[(ii)] $\ttheta_{\veps}^{*}=\theta_{0}+a_{n}^{-1}\beta_{0}A(\theta_{0})^{1/2}w_{\rm obs}+\veps_{n}(\beta_{0}-\beta_{\veps})E_{G_{n},r}(\nu)+r_{2}$,
where $r_{2}=o_{P}(a_{n}^{-1})$. 
\end{enumerate}
\end{lemma}
\begin{proof}
These results generalize Lemma 3(c) and Lemma 5(c) in \cite{Li2017}.
With the transformation $t=t(\theta)$, by Lemma 2, if $\veps_{n}=o(a_{n}^{-1/2})$,
\begin{eqnarray}
\begin{cases}
\ttheta_{\veps}=\theta_{0}+a_{n,\veps}^{-1}\tilde{q}_{B_{\delta},t\nu}(t)/\tilde{q}_{B_{\delta},t\nu}(1)=\theta_{0}+a_{n,\veps}^{-1}g_{B_{\delta},r}(t)/g_{B_{\delta},r}(1)+o_{p}(a_{n}^{-1}),\\
\ttheta_{\veps}^{x}=\theta_{0}+a_{n,\veps}^{-1}\tilde{q}_{B_{\delta},t\nu}(t)/\tilde{q}_{B_{\delta},t\nu}(1)-\veps_{n}\beta_{\veps}\tilde{q}_{B_{\delta},t\nu}(\nu)/\tilde{q}_{B_{\delta},t\nu}(1)\\
\hspace{5mm} =\theta_{0}+a_{n,\veps}^{-1}g_{B_{\delta},r}(t)/g_{B_{\delta},r}(1)-\veps_{n}\beta_{\veps}E_{a_{n},r}(\nu)+o_{p}(a_{n}^{-1}),
\end{cases}\label{eq4}
\end{eqnarray}
where the remainder term comes from the fact that $(a_{n,\veps}^{-1}+\veps_{n})\left\{ O_{p}(a_{n,\veps}^{-1})+O_{p}(a_{n}^{2}\veps_{n}^{4})\right\} =o_{p}(a_{n}^{-1})$.

First the leading term of $g_{B_{\delta},r}(t\nu^{k})$ is derived
for $k=0$ or $1$. The case of $k=1$ will be used later. Let $t'=t-\psi(\nu)$,
then 
\begin{align*}
g_{B_{\delta},r}(t\nu^{k_{2}}) & =\int_{\mathbb{R}^{d}}\int_{t(B_{\delta})}\{t'+\psi(\nu)\}\nu^{k_{2}}\gamma_{n}(t)g_{n}(t,\nu)\,dtd\nu\\
 & =\int_{\mathbb{R}^{d}}\psi(\nu)\nu^{k_{2}}G_{n,r}(\nu)\,d\nu+\int_{\mathbb{R}^{d}}\int_{t(B_{\delta})}t'\nu^{k_{2}}\gamma_{n}(t)g_{n}(t,\nu)\,dtd\nu.
\end{align*}
By matrix algebra, it is straightforward to show that $g_{n}(t,v)=N\{t;\psi(v),k_{n}^{-2}I(\theta_{0})^{-1}\}G_{n}(v)$.
Then with the transformation $t'$, we have
\begin{align*}
 & g_{B_{\delta},r}(t\nu^{k_{2}})-\int_{\mathbb{R}^{d}}\psi(\nu)\nu^{k_{2}}G_{n,r}(\nu)\,d\nu\\
= & \int_{\mathbb{R}^{d}}\int_{t(B_{\delta})-\psi(\nu)}t'\nu^{k_{2}}\gamma_{n}\{\psi(\nu)+t'\}N\left\{ t';0,k_{n}^{-2}I(\theta_{0})^{-1}\right\} G_{n}(\nu)\,dt'd\nu.
\end{align*}
By applying the Taylor expansion on $\gamma_{n}\{\psi(\nu)+t'\}$,
the right hand side of the above equation is equal to 
\begin{eqnarray}
 &  & \int_{\mathbb{R}^{d}}\int_{t(B_{\delta})-\psi(\nu)}t'N\{t';0,k_{n}^{-2}I(\theta_{0})^{-1}\}\,dt'\cdot\gamma_{n}\{\psi(\nu)\}\nu^{k_{2}}G_{n}(\nu)\,d\nu\nonumber \\
 &  & +\int_{\mathbb{R}^{d}}\int_{t(B_{\delta})-\psi(\nu)}t'{}^{2}D_{t}\gamma_{n}\{\psi(\nu)+e_{t}\}N\{t';0,k_{n}^{-2}I(\theta_{0})^{-1}\}\,dt'\cdot\nu^{k_{2}}G_{n}(\nu)d\nu\nonumber \\
 & = & k_{n}^{-1}\int_{\mathbb{R}^{d}}\int_{Q_{v}}t''N\{t'';0,I(\theta_{0})^{-1}\}\,dt''\cdot\gamma_{n}\{\psi(\nu)\}\nu^{k_{2}}G_{n}(\nu)\,d\nu\nonumber \\
 &  & +k_{n}^{-2}\int_{\mathbb{R}^{d}}\int_{Q_{v}}t''^{2}D_{t}\gamma_{n}\{\psi(\nu)+e_{t}\}N\{t'';0,I(\theta_{0})^{-1}\}\,dt''\cdot\nu^{k_{2}}G_{n}(\nu)\,d\nu,\label{eq5}
\end{eqnarray}
where $Q_{v}=\left\{ a_{n}(\theta-\theta_{0})-k_{n}\psi(\nu)\mid\theta\in B_{\delta}\right\} $
and $t''=k_{n}t'$. Since $Q_{v}$ can be written as $\left\{ a_{n}(\theta-\theta_{0}-\veps_{n}\nu)-\beta_{0}A(\theta_{0})^{1/2}W_{\rm obs}\mid\theta\in B_{\delta}\right\} $,
it converges to $\mathbb{R}^{p}$ for any fixed $v$ with probability
one. Then $\int_{Q_{v}}t''N\{t'';0,\tau(\theta_{0})^{-1}\}\,dt''=o_{P}(1)$
for fixed $v$, and by the continuous mapping theorem and Condition \ref{initial_upper}, the
first term in the right hand side of \eqref{eq5} is of the order
$o_{p}(k_{n}^{-1})$. The second term is bounded by 
\[
k_{n}^{-2}\sup_{t\in\mathbb{R}}\|D_{t}\gamma_{n}(t)\|\int_{\mathbb{R}^{p}}\|t''\|^{2}N\{t'';0,I(\theta_{0}^{-1})\}\,dt''\int_{\mathbb{R}^{d}}\nu^{k_{2}}G_{n}(\nu)\,d\nu,
\]
which is of the order $O_{p}(k^{-2}\tau_{n}/a_{n,\veps})$ by Condition \ref{initial_gradient}.
Therefore 
\begin{align}
g_{B_{\delta},r}(t\nu^{k_{2}}) & =\int_{\mathbb{R}^{d}}\psi(\nu)\nu^{k_{2}}G_{n}(\nu)d\nu+o_{P}(k_{n}^{-1}).\label{eq6}
\end{align}
By algebra, $k_{n}=a_{n,\veps}^{-1}a_{n}$, and 
\begin{eqnarray}
 &  & \int_{\mathbb{R}^{d}}\psi(\nu)\nu^{k_{2}}G_{n}(\nu)d\nu\nonumber \\
 & = & a_{n,\veps}\beta_{0}\{a_{n}^{-1}A(\theta_{0})^{1/2}W_{\rm obs}\int_{\mathbb{R}^{d}}\nu^{k_{2}}G_{n,r}(\nu)\,d\nu+\veps_{n}\int_{\mathbb{R}^{d}}\nu^{k_{2}+1}G_{n,r}(\nu)\,d\nu\}.\label{eq7}
\end{eqnarray}
Then ($i$) and ($ii$) in the Lemma holds by plugging the expansion
of $g_{B_{\delta},r}(t)$ into \eqref{eq4}.
\end{proof}
\begin{lemma}\label{Alemma3.5} Assume Condition \ref{par_true}, \ref{initial_upper}, \ref{sum_conv}--\ref{sum_approx_tail}. Then as $n\rightarrow\infty$, 
\begin{enumerate}
\item[(i)] For any $\delta<\delta_{0}$, $r_{B_{\delta}^{c}}(1)$ and $\tilde{q}_{B_{\delta}^{c}}(1)$
are $o_{p}(\tau_{n}^{p})$. More specifically, they are of the order
$O_{p}\left(\tau_{n}^{p}e^{-a_{n,\veps}^{\alpha_{\delta}}c_{\delta}}\right)$
for some positive constants $c_{\delta}$ and $\alpha_{\delta}$ depending
on $\delta$.
\item[(ii)] $q_{B_{\delta}}(1)=\tilde{q}_{B_{\delta}}(1)\{1+O_{p}(\alpha_{n}^{-1})\}$
and $\sup_{A\subset B_{\delta}}|q_{A}(1)-\tilde{q}_{A}(1)|/\tilde{q}_{B_{\delta}}(1)=O_{p}(\alpha_{n}^{-1})$; 
\item[(iii)] if $\veps_{n}=o(a_{n}^{-1/2})$, then $\tilde{q}_{B_{\delta}}(1)$ and
$r_{B_{\delta}}(1)$ are $\Theta_{P}(\tau_{n}^{p}a_{n,\veps}^{d-p})$,
and thus $\tilde{q}_{\mathcal{P}_{0}}(1)$ and $q_{\mathcal{P}_{0}}(1)$
are $\Theta_{P}(\tau_{n}^{p}a_{n,\veps}^{d-p})$; 
\item[(iv)] if $\veps_{n}=o(a_{n}^{-1/2})$, $\theta_{\veps}=\ttheta_{\veps}+o_{p}(a_{n,\veps}^{-1}).$
If $\veps_{n}=o(a_{n}^{-3/5}),$ $\theta_{\veps}=\ttheta_{\veps}+o_{P}(a_{n}^{-1}).$
\end{enumerate}
\end{lemma}
\begin{proof}
This generalizes Lemma 7 in \cite{Li2017}. The arguments therein
can be followed exactly, by Condition \ref{initial_upper} and the fact that regarding $\pi(\theta)$,
only the condition $\sup_{\theta\in\mathbb{R}^{p}}\pi(\theta)<\infty$
is used.
\end{proof}
\begin{lemma}\label{Alemma4} Assume Condition \ref{par_true}, \ref{initial_upper}, \ref{sum_conv}--\ref{sum_approx_tail}. 
\begin{enumerate}
\item[(i)] For any $\delta<\delta_{0}$, $Q_{\veps}(\theta\in B_{\delta}^{c}\mid s_{\rm obs})$
and $\tilde{Q}_{\veps}(\theta\in B_{\delta}^{c}\mid s_{\rm obs})$ are $o_{p}(1)$; 
\item[(ii)] There exists some $\delta<\delta_{0}$ such that 
\[
\sup_{A\in\mathfrak{B}^{p}}|Q_{\veps}(\theta\in A\cap B_{\delta}\mid s_{\rm obs})-\tilde{Q}_{\veps}(\theta\in A\cap B_{\delta}\mid s_{\rm obs})|=o_{p}(1);
\]
\item[(iii)] $a_{n,\veps}(\theta_{\veps}-\ttheta_{\veps})=o_{p}(1)$ . 
\end{enumerate}
\end{lemma}
\begin{proof}
This lemma generalizes Lemma 3 of \cite{Li2017}. The proof of Lemma
3 in \cite{Li2017} only needs Lemma 3 and 5 in \cite{Li2016} to
hold. The result that $q_{B_{\delta}^{c}}\{h(\theta)\}=O_{p}(\tau_{n}^{p}e^{-a_{n,\veps}^{\alpha_{\delta}}c_{\delta}})$
for some positive constants $\alpha_{\delta}$ and $c_{\delta}$,
which generalizes the case of $r_{n}(\theta)=\pi(\theta)$ in Lemma
3 of \cite{Li2016}, holds by Condition \ref{initial_upper}, since for the latter, regarding
$\pi(\theta)$ it only uses the fact that $\sup_{\theta\in B_{\delta}^{c}}\pi(\theta)<\infty$.
Then the arguments in the proof of Lemma 3 in \cite{Li2016} can be
followed exactly, despite the term $\tau_{n}^{p}$ that is not included
in the order of $\pi_{B_{\delta}^{c}}\{h(\theta)\}$, since $Q_{\veps}(\theta\in A\mid s_{{\rm obs}})$
is the ratio $q_{A}(1)/q_{\mathbb{R}^{p}}(1)$. Since Lemma 5 in \cite{Li2016}
has been generalized by Lemma \eqref{Alemma2}, the arguments of the
proof of Lemma 3 in \cite{Li2017} can be followed exactly.
\end{proof}
\textbf{Proof of Theorem 2:} This result generalizes the case ($i$)
of Proposition 1 in \cite{Li2017}. With the above lemmas, lines for
proving case $(i)$ of Proposition 1 in \cite{Li2017} can be followed
exactly.

\subsection*{Proof of Theorem 3}

\begin{lemma}\label{Alemma5} Assume Condition \ref{par_true}--\ref{cond:likelihood_moments}. If $\veps_{n}=o_{p}(a_{n}^{-3/5})$, then $a_{n}\veps_{n}(\beta_{\veps}-\beta_{0})=o(1)$.
\end{lemma}
\begin{proof}
This generalizes Lemma 4 in \cite{Li2017} by replacing $\pi(\theta_{0}+a_{n,\veps}^{-1}t)$
therein with $\gamma_{n}(t)$. By Condition \ref{initial_upper} and the arguments in the proof
of Lemma 4 in \cite{Li2017}, it can be shown that 
\[
\frac{q_{\mathbb{R}^{p}}\{(\theta-\theta_{0})^{k_{1}}(s-s_{\rm obs})^{k_{2}}\}}{q_{\mathbb{R}^{p}}(1)}=a_{n,\veps}^{-k_{1}}\veps_{n}^{-k_{2}}\left\{ \frac{\tilde{q}_{B_{\delta},tv}(t^{k_{1}}\nu^{k_{2}})}{\tilde{q}_{B_{\delta},tv}(1)}+O_{p}(\alpha_{n}^{-1})\right\} .
\]
Then by Lemma 2 $(iii)$, the right hand side of the above is equal
to 
\[
a_{n,\veps}^{-k_{1}}\veps_{n}^{-k_{2}}\left\{ \frac{g_{B_{\delta},r}(t^{k_{1}}\nu^{k_{2}})}{g_{B_{\delta},r}(1)}+O_{p}(a_{n,\veps}^{-1})+O_{p}(a_{n}^{2}\veps_{n}^{4})+O_{p}(\alpha_{n}^{-1})\right\} .
\]
Since $\beta_{\veps}=\text{Cov}_{\veps}(\theta,S_{n})\text{Var}_{\veps}(S_{n})^{-1}$,
\textcolor{red}{{} }
\begin{align*}
a_{n}\veps_{n}(\beta_{\veps}-\beta_{0})= & k_{n}\left[\frac{g_{B_{\delta},r}(t\nu)}{g_{B_{\delta},r}(1)}-\frac{g_{B_{\delta},r}(t)g_{B_{\delta},r}(\nu)}{g_{B_{\delta},r}(1)^{2}}+o_{p}(k_{n}^{-1})\right]\cdot\\
 & \qquad\left[\frac{g_{B_{\delta},r}(\nu\nu^{T})}{g_{B_{\delta},r}(1)}-\frac{g_{B_{\delta},r}(\nu)g_{B_{\delta},r}(\nu)^{T}}{g_{B_{\delta},r}(1)^{2}}+o_{p}(k_{n}^{-1})\right]-a_{n}\veps_{n}\beta_{0},
\end{align*}
 where the equations that $a_{n,\veps}^{-1}k_{n}=o(1)$, $a_{n}^{2}\veps_{n}^{4}k_{n}=o(p)$,
and $\alpha_{n}^{-1}k_{n}=o(a_{n}^{-2/5}k_{n})=o(1)$ are used. By
algebra, the right hand side of the equation above can be rewritten
as 
\begin{eqnarray*}
 &  & \left\{ \frac{g_{B_{\delta},r}\{(k_{n}t-a_{n}\veps_{n}\beta_{0}\nu)\nu\}}{g_{B_{\delta},r}(1)}-\frac{g_{B_{\delta},r}(k_{n}t-a_{n}\veps_{n}\beta_{0}\nu)g_{B_{\delta},r}(\nu)}{g_{B_{\delta},r}(1)^{2}}+o_{p}(1)\right\} \cdot\\
 &  & \qquad\left\{ E_{G,r}(\nu\nu^{T})-E_{G,r}(\nu)E_{G,r}(\nu)^{T}+o_{p}(k_{n}^{-1})\right\} ^{-1}.
\end{eqnarray*}
By plugging \eqref{eq6} and \eqref{eq7} in the above, $a_{n}\veps_{n}(\beta_{\veps}-\beta_{0})$
is equal to 
\begin{eqnarray*}
 &  & \left\{ E_{G,r}(\nu)\beta_{0}A(\theta_{0})^{1/2}W_{\rm obs}-E_{G,r}(\nu)\beta_{0}A(\theta_{0})^{1/2}W_{\rm obs}+o_{p}(1)\right\} \cdot%\\
 %&  & \hspace{1cm}
\{\text{Var}_{G,r}(\nu)+o_{p}(k_{n}^{-1})\}^{-1}\\
 &&\hspace{1cm} = o_{P}(1)\{\text{Var}_{G,r}(\nu)+o_{p}(k_{n}^{-1})\}^{-1}.
\end{eqnarray*}
Since
\begin{align*}
\text{Var}_{G,r}(\nu) & \geq\frac{\inf_{t\in t(B_{\delta'})}\gamma_{n}(t)}{g_{B_{\delta},r}(1)}\int_{\mathbb{R}^{d}}\int_{t(B_{\delta'})}\{\nu-E_{G,r}(\nu)\}^{2}g_{n}(t,\nu)\,dtd\nu,
\end{align*}
where $\delta'$ is defined in the proof of Lemma \ref{Alemma2}(ii),
we have $\text{Var}_{G,r}(\nu)^{-1}=\Theta_{p}(1)$. Therefore $a_{n}\veps_{n}(\beta_{\veps}-\beta_{0})=o_{p}(1)$.
\end{proof}
\begin{lemma}\label{Alemma6} Results generalizing Lemma 5 in \cite{Li2017},
$i.e.$ replacing $\Pi_{\veps}$ and $\tPi_{\veps}$ therein with
$Q_{\veps}$ and $\tilde{Q}_{\veps}$, hold. \end{lemma} 
\begin{proof}
In \cite{Li2017}, the proof of Lemma 5 requires Lemma 4 and 7 in
\cite{Li2017} to hold. Since their generalized results have been
proved, the proof of this lemma follows the same arguments.
\end{proof}
\begin{lemma}\label{Alemma7} Results generalizing Lemma 10 in \cite{Li2017}
hold. \end{lemma}
\begin{proof}
The same arguments can be followed.
\end{proof}
\noindent \textbf{Proof of Theorem 3:} With all above lemmas, the
proof holds by following the same arguments in the proof of Theorem
1 in \cite{Li2017}.

	%------------------------------------------------------------------------------------------------------------
	\bibliographystyle{biometrika}
	\bibliography{ACC_paper_2018_11_19}
	
	%\begin{thebibliography}{7}
	%	\expandafter\ifx\csname natexlab\endcsname\relax\def\natexlab#1{#1}\fi	
	%\end{thebibliography}
	
\end{document}

%% file: Symbols_pream.tex
\renewcommand{\pmb}{{}}

\def\benu{\begin{enumerate}}
	\def\eenu{\end{enumerate}}

\def\real{{\mathbb{R}}}
\def\R{{\real}}

\def\veps{\varepsilon}

\def\htheta{\hat{\theta}_{\text{S}}}
\def\bsob{\pmb{s}_{\rm obs}}

\def\r{r_n}								% r_n
\def\P{{\cal P}}									% parameter space

\def\ftil{\widetilde{f}}

\def\bs{\pmb{s}}

\def\bx{\pmb{x}}

\def\bS{\pmb{S}}

\def\bmu{\pmb{\mu}}

\def\tPi{\widetilde{\Pi}}

\def\btheta{\pmb{\theta}}

\def\ttheta{\widetilde{\theta}}